\documentclass[letterpaper,11pt]{article}
\usepackage{chihao}

\newcommand{\1}[2][]{\mathbbm{1}^{#1}\left[#2\right]}

\usepackage[margin=1in]{geometry}
\usepackage[affil-it]{authblk}

\newcommand{\DTV}{D_{\!{TV}}}

\title{Sampling Proper Colorings on Line Graphs Using $(1+o(1))\Delta$ Colors}
\author[1]{Yulin Wang\thanks{yulin\_wang@sjtu.edu.cn}}
\author[2]{Chihao Zhang\thanks{chihao@sjtu.edu.cn}}
\author[3]{Zihan Zhang\thanks{zihan@nii.ac.jp. The work was done while Zihan Zhang was an undergraduate student at Shanghai Jiao Tong University.}}
\affil[1,2]{Shanghai Jiao Tong University}
\affil[3]{Graduate Institute for Advanced Studies, SOKENDAI}

\begin{document}
\maketitle

\begin{abstract}
We prove that the single-site Glauber dynamics for sampling proper $q$-colorings mixes in $O_\Delta(n\log n)$ time on line graphs with $n$ vertices and maximum degree $\Delta$ when $q>(1+o(1))\Delta$. The main tool in our proof is the matrix trickle-down theorem developed by Abdolazimi, Liu and Oveis Gharan in~\cite{ALOG21}. 
\end{abstract}
\section{Introduction}

A proper (vertex) $q$-coloring of a graph $G=(V,E)$ is an assignment of each vertex to one of $q$ colors so that the colors of adjacent vertices are distinct. Let $\Delta$ be the maximum degree of $G$. It has been widely conjectured that the single-site Glauber dynamics for uniformly sampling proper colorings on $G$ mixes rapidly as long as the number of colors $q$ is at least $\Delta+2$. Since the seminal work of Jerrum~\cite{Jer95} and Salas and Sokal~\cite{SS97} where Glauber dynamics was shown to be rapidly mixing when $q\ge 2\Delta$, a number of works devoted to resolving the conjecture and the current best bound requires $q\ge \tp{\frac{11}{6}-\epsilon_0}\Delta$ for some $\epsilon_0\approx 10^{-5}$~\cite{CDM+19}, which is still far from desired. 

Another line of work to approach the conjecture is by considering special graph families, with the most successful cases being those graphs with large girths (e.g.,~\cite{DF03, Mol04, DFHV13,CGSV21,FGYZ20}). Notably in a very recent work of Chen, Liu, Mani and Moitra~\cite{CLMM23}, it was proven that for any $\Delta\ge 3$, there exists a $\-g>0$ such that for any graph $G$ with the maximum degree at most $\Delta$ and the girth at least $\-g$, the single-site Glauber dynamics for sampling proper $q$-colorings on $G$ mixes rapidly as long as $q\ge \Delta +3$. This result almost solves the aforementioned conjecture regarding the case of large girth graphs.

A closely related problem is \emph{sampling proper edge colorings}, where one assigns colors to the edges instead of vertices so that adjacent edges have distinct colors. It is clear that a proper edge coloring of $G$ naturally corresponds to a vertex coloring of the line graph of $G$ in which the vertices are the edges of $G$ and two vertices are adjacent in the line graph if and only if the corresponding edges are incident in $G$. As a result, studying sampling proper colorings on line graphs is of particular interest. Unlike graphs with a large girth which are ``locally sparse'', line graphs are ``locally dense''. Specifically, a line graph is formed by gluing together several cliques with each vertex belonging to exactly two cliques. Previous methods such as coupling can hardly take advantage of this special structure and therefore the rapid mixing condition for line graphs remains the same as that for general graphs for a long time. 

In a recent breakthrough, Abdolazimi, Liu and Oveis Gharan~\cite{ALOG21} developed a new technique to establish rapid mixing results for Glauber dynamics, namely the matrix trickle-down theorem, which can well-utilize the structural information of the underlying graphs. As a result, it was shown in~\cite{ALOG21} the Glauber dynamics for sampling colorings on line graphs mixes rapidly as long as $q > \frac{10}{6}\Delta$. The technique, incorporating with recent revolutionary development of the ``local-to-global'' scheme for high-dimensional expanders~\cite{AL20,CLV21}, yields optimal mixing time for certain Markov chains.  Specifically, the matrix trickle-down theorem is a generalization of the trickle-down theorem in~\cite{Opp18} which established connections between the spectral gaps of the local walks on links of the underlying simplicial complexes. Instead of solely looking at the spectral gap, the matrix trickle-down theorem takes into account the local walks themselves, and establishes more general connections between (the transition matrix of) local walks on links. 

To apply the matrix trickle-down, one needs to design appropriate matrix upper bounds for local walks while keeping their spectrum bounded. Moreover, these matrices must satisfy certain inductive inequality constraints. The construction of these matrix upper bounds is the main technical challenge since one has to have precise control over the magnitudes of their eigenvalues. In this work, we systematically study the constraints arising in the matrix trickle-down theorem and provide an almost optimal construction of matrix upper bounds for sampling proper colorings on line graphs. Our construction has several advantages compared to the one in~\cite{ALOG21}:
\begin{itemize}
	\item Our construction of the matrix upper bounds for each local walk is explicit.
	\item We directly relate the spectrum of each matrix upper bound to that of the adjacent matrix of the line graph, which is a clique locally and is therefore well-understood. Therefore, we can obtain desired spectrum bound in an improved regime.
	\item We reduce the existence of matrix upper bounds to the feasibility of a system of inequalities, for which we give a complete characterization.
\end{itemize}
As a result, we obtain an $O_\Delta(n \log n)$\footnote{The notation $O_\Delta(f)$ means that the hidden constant before $f$ might depend on $\Delta$.} mixing time for bounded degree line graph colorings with $o(\Delta)$ extra colors. 
\begin{theorem}
\label{thm:main-informal-1}
Let $G = (V, E)$ be a line graph with $n$ vertices and maximum degree $\Delta$. If $q\ge \Delta + 20028\cdot \Delta / \log\Delta$, then the Glauber dynamics on
the $q$-colorings of $G$ has modified log-Sobolev constant $\Omega_\Delta(1/n)$, and thus mixes in time $O_\Delta(n\log n)$.
\end{theorem}
Moreover, we obtain an $O(n^{19/9} \log q)$ mixing time on the (degree unbounded) family of line graphs with $o(\Delta)$ extra colors.
\begin{theorem}
\label{thm:main-informal-2}
Let $G = (V, E)$ be a line graph with $n$ vertices and maximum degree $\Delta$. If $q\ge \Delta + 20028\cdot \Delta / \log\Delta$, then the Glauber dynamics on
the $q$-colorings of $G$ has spectral gap larger than $n^{-10/9}$, and thus
mixes in time $O(n^{19/9}\log q)$.
\end{theorem}

In fact, our proofs hold for more general \emph{list coloring instances} with $o(\Delta)$ extra colors. See \Cref{thm:mains} for the most general statement.

Our results can also be stated in terms of sampling edge colorings. Then the rapid mixing condition becomes to $q\ge (2+o(1))d$ where $d$ is the maximum degree of the graph (note that the maximum degree of corresponding line graph can be as large as $2d-2$). We remark that this is close to the ergodicity threshold for the single-site Glaubder dynamics for sampling edge colorings since it is known that the chain is reducible when $q<2d$~\cite{MJNP19}. On the other hand, Vizing's theorem states that the graph is edge colorable whenever $q\ge d+1$ and rapid mixing under the same condition is known for trees~\cite{DHP20}.

\bigskip
We will introduce necessary notations and some preliminary results in \Cref{sec:pre}. In particular, we review the matrix trickle-down theorem in \Cref{sec:trickle-down}. Then we describe our construction and prove the main results in \Cref{sec:vertex-coloring}. A main ingredient of our proof, the analysis of the inequalities arising in the matrix trickle-down theorem, is presented in \Cref{sec:solution}. This analysis might be of independent interest.

\section{Preliminaries}
\label{sec:pre}

\subsection{Simplicial Complexes}
Let $U$ be a universe. A \emph{simplicial complex} $\+C\subseteq 2^U$ is a collection of subsets of $U$ that is closed under taking subsets. That is, if $\sigma\in \+C$ and $\sigma'\subseteq\sigma$, then $\sigma'\in \+C$. Every $\sigma\in \+C$ is called a \emph{face}, and a face that is not a proper subset of any other face is called a \emph{maximal face} or a \emph{facet}. The dimension of a face $\sigma$ is $\!{dim}(\sigma)\defeq \abs{\sigma}$, namely the number of elements in $\sigma$. For every $ k\ge 0$, we use $\+C_k\defeq \set{\sigma\in\+C\cmid \abs{\sigma}=k}$ to denote the set of faces of dimension $k$. Specifically, $\+C_0=\set{\emptyset}$. The dimension of $\+C$ is the maximum dimension of faces in $\+C$. We say $\+C$ is a \emph{pure} $n$-dimensional simplicial complex if all maximal faces in $\+C$ are of dimension $n$. In the following, we assume $\+C$ is a pure $n$-dimensional simplicial complex. For every face $\sigma\in \+C$, we define its co-dimension $\!{codim}(\sigma)\defeq n-\!{dim}(\sigma)$.


Let $\pi_n$ be a distribution over the maximal faces $\+C_n$. We use the pair $(\+C,\pi_n)$ to denote a \emph{weighted simplicial complex} where for each $1\le k < n$, the distribution $\pi_n$ induces a distribution $\pi_k$ over $\+C_k$. Formally, for every $1\le k<n$ and every $\sigma'\in \+C_k$,
\begin{align*}
    \pi_k(\sigma') \defeq \frac{1}{\binom n k}\sum_{\sigma\in \+C_n\cmid\sigma' \subset \sigma} \pi_{n}(\sigma).
\end{align*}
One can easily verify that $\pi_k$ is a distribution on $\+C_k$. Combined with $\pi_n$ itself, for each $1\le k\le n$, the distribution over $\+C_k$ is defined. Sometimes, we omit the subscript when $j=1$, i.e., we write $\pi$ for $\pi_1$.

For a face $\tau \in \+C$ of dimension $k$, we define its \emph{link} as
\begin{align*}
    \+C_\tau=\set{\sigma\setminus \tau \cmid \sigma \in \+C \land \tau \subseteq \sigma }.
\end{align*}
Clearly, $\+C_\tau$ is a pure $(n-k)$-dimensional simplicial complex. Similarly, for every $1\le j \le n-k$, we use $\+C_{\tau,j}$ to denote the set of faces in $\+C_\tau$ of dimension $j$. We also use $\pi_{\tau,j}$ to denote the \emph{marginal distribution} on $\+C_{\tau,j}$. Formally, for every $\sigma\in \+C_{\tau,j}$,
\begin{align*}
    \pi_{\tau,j}(\sigma)\defeq \Pr[\alpha \sim \pi_{k+j}]{\alpha=\tau \cup \sigma \cmid \tau\subseteq \alpha}=\frac{\pi_{k+j}(\tau \cup \sigma)}{\binom{k+j}{k} \cdot \pi_{k}(\tau)}.
\end{align*}
We drop the subscript when $j=1$, i.e., we write $\pi_\tau$ for $\pi_{\tau,1}$.
Note that the marginal distributions $\pi_{\tau, j}$ are the same as the distributions over $C_{\tau, j}$ induced from the weighted simplicial complex $\tp{\+C_{\tau}, \pi_{\tau, m-k}}$, so there is no ambiguity about this notation.

Let $(\+C^{(1)},\pi^{(1)})$, $(\+C^{(2)},\pi^{(2)})$ be two pure weighted simplicial complexes of dimension $d_1$ and $d_2$ respectively. We can define another pure weighted simplicial complex $(\+C,\pi)$ of dimension $d_1+d_2$ whose maximal faces are the disjoint union of maximal faces of $\+C^{(1)}$ and $\+C^{(2)}$. Moreover, $\pi$ is the product measure $\pi^{(1)} \times \pi^{(2)}$. $(\+C,\pi)$ is also called the \emph{product} of $(\+C^{(1)},\pi^{(1)})$ and $(\+C^{(2)},\pi^{(2)})$. This definition can be naturally generalized to the products of more than two weighted simplicial complexes.

Then we define notations for matrices related to $\pi_\tau$. Define $\Pi_\tau \in \bb R^{\+C_1\times\+C_1}$ as $\Pi_\tau \triangleq \operatorname{diag}(\pi_\tau)$ supported on $\+C_{\tau,1}\times \+C_{\tau,1}$. For convenience, define the pseudo inverse $\Pi_\tau^{-1} \in \bb R^{\+C_1\times\+C_1}$ of $\Pi_\tau$ as $\Pi_\tau^{-1}(x,x)=\pi_\tau(x)^{-1}$ for $x\in \+C_{\tau,1}$ and $0$ otherwise. Similarly, the pseudo inverse square root $\Pi_\tau^{-1/2} \in \bb R^{\+C_1\times\+C_1}$is defined as $\Pi_\tau^{-1/2}(x,x)=\pi_\tau(x)^{-1/2}$ for $x\in \+C_{\tau,1}$ and $0$ otherwise. 

\subsection{Vertex Coloring}
Fix a color set $[q]=\set{1, 2, \dots, q}$ where $q\in\bb N$. Let $G=(V,E)$ be an undirected graph and $\+L=\set{L_v\subseteq[q]\cmid v\in V}$ be a collection of color lists associated with each vertex in $V$. For every $v\in V$, we use $\ell_v\defeq \abs{L_v}$ to denote the size of $L_v$. The pair $(G, \+L)$  is an instance of list-coloring. If there exists an integer $\beta$ such that $|L_v|>\Delta(v)+\beta$ for any $v\in V$ where $\Delta(v)$ is the degree of $v$, we call $(G,\+L)$ a $\beta$-extra list-coloring instance.

We say $\sigma:V\rightarrow [q]$ is a proper coloring if $\sigma(v) \in L_v$ for any $v\in V$ and $\sigma(u)\neq \sigma(v)$ for any $\set{u, v} \in E$. We also regard $\sigma$ as a set of pairs of vertex and color, namely $\set{(v,c)\in V\times [q]\cmid \sigma(v)=c}$. Let $\Omega$ denote the set of all proper colorings and $\mu$ be the uniform distribution on $\Omega$. Let $\Lambda\subseteq V$ and $\tau\in [q]^\Lambda$. We say $\tau$ is a proper partial coloring on $\Lambda$ if it is a proper coloring on $(G[\Lambda],\+L|_\Lambda)$ where $G[\Lambda]$ is the subgraph of $G$ induced by $\Lambda$ and $\+L|_{\Lambda}=\set{L_v\in \+L\cmid v\in \Lambda}$. We also define $\mu^\tau$ on $\Omega$ as $\mu^\tau(\cdot) = \Pr[\sigma \sim \mu]{\cdot \subset \sigma \mid\tau\subset \sigma}$.
For a subset $S\subseteq V\setminus \Lambda$ and a partial coloring $\omega$ on $S$, define $\mu^\tau_S(\omega)=\Pr[\sigma\sim \mu]{\omega \subset \sigma \mid \tau \subset \sigma}$. 

Assume $\abs{V}=n$. The list-coloring instance $(G,\+L)$ can be naturally represented as a weighted simplicial complex $(\+C,\pi_n)$ where $\+C$ consists of all proper partial colorings and $\pi_n=\mu$.

The following identity is useful throughout the paper:
\begin{proposition}
	Let $\Lambda\subseteq V$. For every partial coloring $\tau$ on $\Lambda$, every $S\subseteq V\setminus \Lambda$ and partial coloring $\omega\in [q]^S$, it holds that
	\[
		\mu^\tau_S(\omega)=\binom{\!{codim}(\tau)}{|S|}\cdot\pi_{\tau, |S|}(\omega).
	\]
\end{proposition}
\subsection{Markov Chains and Mixing Time}
Let $\Omega$ be a finite discrete state space. Let $P$ be the \emph{transition matrix} of a Markov chain with the stationary distribution $\pi$. We say $P$ is \emph{reversible} with respect to the stationary distribution $\pi$ if it satisfies the detailed balance condition, i.e., for every $x,y \in \Omega$, it holds that
\begin{align*}
    \pi(x)P(x,y)=\pi(y)P(y,x).
\end{align*}
Only reversible chains are considered in this paper.

For a weighted simplicial complex $\tp{\+C, \pi_n}$, the \emph{single-site Glauber dynamics} on $\tp{\+C, \pi_n}$ is a Markov chain on $\+C_n$ with the transition matrix
\begin{equation*}
	P_{\!{GL}}\tp{\sigma, \sigma'} \defeq
	\begin{cases}
		\frac 1 n \sum_{x\in \sigma} \pi_{\sigma\setminus\set x, 1} \tp{\set x} & \text{if } \sigma = \sigma',\\
		\frac 1 n \pi_{\sigma\cap\sigma', 1} \tp{\sigma'\setminus\sigma} & \text{if } \sigma\cap\sigma' \in \+C_{n-1}, \\
		0 & \text{otherwise}.
	\end{cases}
\end{equation*}
From an operational view, each transition of the Glauber dynamics, with the current state being $\sigma$, consists of two steps:
\begin{enumerate}
	\item Uniformly select a random $x\in \sigma$.
	\item Select a $y\in \+C_{\sigma\setminus \set x, 1}$ following the distribution $\pi_{\sigma\setminus \set x, 1}$ and transfer to the state $\sigma\setminus\set x\cup y$.
\end{enumerate}
One can easily verify that the Glauber dynamics on $\tp{\+C, \pi_n}$ is reversible, with $\pi_n$ as the stationary distribution.

We are concerned with the convergence rate of Markov chains, which is described by the mixing time.
The mixing time is defined as the duration required
for the total variation distance between $\mu_t\defeq(P^t)^{\top}\mu_0$
and the stationary distribution $\pi$ to become smaller than $\eps$, starting from any initial distribution $\mu_0$. Formally,
\begin{align*}
    t_{\!{mix}}(\eps)\defeq\min_{t>0} \max_{\mu_0} \DTV(\mu_t,\pi)\leq \eps,
\end{align*}
where $\DTV(\mu, \nu)\defeq \frac{1}{2}\sum_{x\in \Omega}\abs{\mu(x)-\nu(x)}$ is the total variation distance.


For a reversible Markov chain $P$ on a discrete space $\Omega$, since $P$ is self-adjoint with respect to the inner product induced by $\pi$, all eigenvalues of $P$ are real.
So we can define the \emph{spectral gap} of $P$ as $\lambda(P)\defeq 1-\lambda_2(P)$, where $\lambda_2(P)$ denotes the second largest eigenvalue of $P$.
And the \emph{absolute spectral gap} of $P$ as $\lambda_\star\defeq 1 - \max\set{\vert\lambda\vert\cmid \lambda\text{ is an eigenvalue of } P \text{ and } \lambda\neq 1}$.
The following lemma arises to bound the mixing time of a reversible Markov chain by its absolute spectral gap.
\begin{lemma}[Theorem 12.4 in \cite{LPW17}]
	\label{lem:sg to mixing}
	For an irreducible reversible Markov chain $P$ on a discrete space,
	\begin{align}
		t_{\!{mix}}(P, \eps)\leq \frac{1}{\lambda_\star(P)}\left(\frac{1}{2} \log \frac{1}{\pi_{\min}}+ \log\frac{1}{2\eps}\right),
	\end{align}
where $\pi_{\min} = \min_{x\in \Omega} \pi(x)$.
\end{lemma}
Notice that when all eigenvalues of $P$ are non-negative, the absolute spectral gap in the above lemma equals the spectral gap. This is the case of Glauber dynamics, as in \Cref{prop:ltog-spectral}.

\subsection{Local-to-Global Scheme}
The \emph{local random walk} $P_\tau$ on $\+C_{\tau,1}$ is defined as
\begin{align*}
    P_\tau(x,y)=\frac{\pi_{\tau,2}(\{x,y\})}{2\pi_{\tau,1}(x)},
\end{align*}
for all $x, y \in \+C_{\tau,1}$. An operational view of the local chain is as follows: when the current state is at $x\in \+C_{\tau,1}$, move to $y\in \+C_{\tau,1}$ with probability proportional to $\pi_{\tau,2}(\set{x,y})$. Note that we will treat $P_\tau$ as a matrix in $[0,1]^{\+C_1\times\+C_1}$ such that the undefined entries are $0$. It is obvious that $P_\tau$ is reversible with respect to $\pi_\tau$. Specifically, we denote the local random walk on $\+C_1$ by $P$, i.e.,
\begin{align*}
    P(x,y)=\frac{\pi_2(\{x,y\})}{2\pi_1(x)}.
\end{align*}
$P$ is also reversible with respect to $\pi$.
We say a weighted simplicial complex $(\+C,\pi_n)$ is $(\gamma_0, \gamma_1,\dots, \gamma_{n-2})$-\emph{local spectral expander} if for any $0\leq k \leq n-2$ and $\tau \in \+C_k$, $\lambda_2(P_\tau)\leq \gamma_k$.

We focus on the spectral gaps of local walks in this paper. As studied earlier by~\cite{AL20}, the local spectral expansion implies bounds for Glauber dynamics. 
\begin{proposition}[\cite{AL20}]
\label{prop:ltog-spectral}
	Let $(\+C, \pi_n)$ be a weighted simplicial complex where $\pi_n$ is a uniform distribution over proper list-colorings over a graph $G=(V,E)$ with $|V|=n$ and maximum degree $\Delta$. If $\tp{\+C, \pi_n}$ is a $(\gamma_0, \gamma_1, \dots, \gamma_{n-2})$-local spectral expander, then
	\begin{enumerate}
		\item all eigenvalues of the Glauber dynamics are real and non-negative;
		\item the second largest eigenvalue of the Glauber dynamics is at most $1 - \frac{1}{n}\prod_{i=0}^{n-2}(1-\gamma_i)$.
	\end{enumerate}
\end{proposition}
The mixing time in terms of the local spectral expansion then follows from the proposition and \Cref{lem:sg to mixing}. To obtain a tighter mixing time bound, we employ the following proposition concerning the local spectral expansion and the modified log-Sobolev constant.

\begin{proposition}[\cite{CLV21}]
\label{prop:ltog}
Let $(\+C,\pi_n)$ be a weighted simplicial complex where $\pi_n$ is a uniform distribution over proper list-colorings of a graph $G=(V,E)$ with $|V|=n$ and maximum degree $\Delta$. If $(\+C,\pi_n)$ is a $(\gamma_0, \gamma_1, \dots, \gamma_{n-2})$-local spectral expander with $\gamma_k \leq \frac{\gamma}{n-k}$ for all $k$, then the modified log-Sobolev constant is at least $\Omega_{\gamma,\Delta}(1/n)$, and the mixing time of the Glauber dynamics is at most $O_{\gamma,\Delta}(n\log n)$.
\end{proposition}

\subsection{Trickle-Down Theorems}\label{sec:trickle-down}

The trickle-down theorem of Oppenheim states that the spectral gaps of local walks in a certain dimension imply spectral gaps of local walks in larger dimensions. 

\begin{proposition}[Trickle-Down Theorem in \cite{Opp18}]\label{prop:trickle-down}
	Given a weighted simplicial complex $(\+C,\pi_d)$, suppose the following holds: 
	\begin{itemize}
	\item $\lambda_2(P)<1$, i.e., the local walk $P$ is irreducible;
	\item There exists some $0\leq \lambda \leq 1/2$ such that $\lambda_2(P_x)\leq \lambda$ for all $x \in \+C_1$.
	\end{itemize}
	Then the local walk $P$ satisfies the spectral bound $\lambda_2(P)\leq \frac{\lambda}{1-\lambda}$.
\end{proposition}


A more general version of the trickle-down theorem was established in~\cite{ALOG21}. Instead of bounding the second largest eigenvalues of local walks, it bounds (the transition matrix of) local walks directly.

A symmetric matrix $A$ is \emph{positive semi-definite}, written as $A\mge 0$ if and only if all its eigenvalues are nonnegative. For two symmetric matrices $A$ and $B$ of the same dimension, we write $A\mge B$, or equivalently $B\mle A$ if and only if $A-B\mge 0$. As a result, the binary relation $\mle$ defines an order between matrices called  \emph{Loewner Order}. For brevity, we write $A\mle_{\pi}B$ if $\Pi A\mle \Pi B$ for any $A,B\in \bb R^{\+C_1\times \+C_1}$ in the following statement.
\begin{proposition}[Theorem 3.2 in \cite{ALOG21}]\label{prop:mtd}
	Given a $d$-dimensional weighted simplicial complex $(\+C,\pi_d)$, suppose the following conditions hold:
	\begin{itemize}
		\item $\lambda_2(P)<1$ where $P$ is the local walk on $\+C_1$;
		\item For a family of matrices $\set{N_x\in \bb R^{\+C_1\times \+C_1}}$ and a constant $\alpha\ge \frac{1}{2}$,
		\begin{equation}\label{eqn:Pxcond}
			P_x-\alpha \*1 \pi_x^{\top} \mle_{\pi_x} N_x \mle_{\pi_x} \frac{1}{2\alpha+1} I,
		\end{equation}
		where $\pi_x$ is the stationary distribution of $P_x$.
	\end{itemize}
	Then for any matrix $N\in\bb R^{\+C_1\times \+C_1}$ satisfying $N\mle_{\pi} \frac{1}{2\alpha} I$ and $\Pi^{-1}\E[x\sim\pi]{\Pi_xN_x}\mle_{\pi} N-\alpha N^2$, it holds that
	\[
		P-\tp{2-\frac{1}{\alpha}}\cdot\*1\pi^{\top} \mle_{\pi} N.
	\]
	In particular, $\lambda_2(P)\le \lambda_1(N)$.
\end{proposition}

We include a proof of \Cref{prop:mtd} in \Cref{sec:proof-mtd} for completeness. 

\bigskip

The following proposition is the main tool we will use to prove the main theorems. It was obtained in~\cite{ALOG21} by applying \Cref{prop:mtd} to simplicial complexes inductively.

%

\begin{proposition}[Theorem 1.3 in \cite{ALOG21}]
	\label{prop:mtd-inductive}
	Given a pure $d$-dimensional weighted simplicial complex $(\+C,\pi_d)$, if there exists a family of matrices $\set{M_\tau\in\bb R^{\+C_1\times \+C_1}}$ satisfying
	\begin{itemize}
		\item For every $\tau\in\+C_{d-2}$, 
		\[
			\Pi_\tau P_\tau -2 \pi_\tau \pi_\tau^\top \mle M_\tau \mle \frac{1}{5}\Pi_\tau;
		\]
		\item For every face $\tau\in \+C_{d-k}$ with $k>3$, one of the following two conditions hold:
		\begin{enumerate}
			\item 
			\[
				M_\tau \mle \frac{k-1}{3k-1}\Pi_\tau\quad\mbox{ and }\quad\E[x\sim\pi_\tau]{M_{\tau\cup\set{x}}}\mle M_\tau -\frac{k-1}{k-2}M_\tau\Pi^{-1}_\tau M_\tau
			\]
			\item $(\+C_\tau,\pi_{\tau,k})$ is the product of $M$ pure weighted simplicial complexes $(\+C^{(1)},\pi^{(1)}), \dots (\+C^{(M)},\pi^{(M)})$ of dimension $n_1,\dots,n_M$ respectively and
			\[
				M_\tau = \sum_{i\in [M]\colon n_i\ge 2} \frac{n_i(n_i-1)}{k(k-1)}\cdot M_{\tau\cup \eta_{-i}}
			\]
			where $\eta_{-i} = \eta\setminus \+C^{(i)}_1$ for an arbitrary $\eta\in \+C_{\tau,k}$.
			\end{enumerate}
	\end{itemize}
	Then for every face $\tau\in\+C_{d-k}$ with $k>2$, it holds that
	\[
		\Pi_\tau P_\tau -\frac{k}{k-1}\pi_\tau \pi_\tau^\top \mle M_\tau \mle \frac{k-1}{3k-1}\Pi_\tau.
	\]
	In particular, $\lambda_2(P_\tau)\le \lambda_1(\Pi_\tau^{-1}M_\tau)$.
\end{proposition}

\subsection{Properties of Loewner Order}

 We collect some useful results on the properties of the Loewner order below.
\begin{lemma}
	\label{lem:matrix-basicineq}
	Let $A, B$ be two matrices in $\bb R^{n \times n}$. For any constant $\eps>0$, we have 
	$AB^\top + BA^\top \mle \eps AA^\top + \frac{1}{\eps} BB^\top$. Moreover, $(A+B)(A+B)^\top \mle (1+\eps)AA^\top + \tp{1+\frac{1}{\eps}}BB^\top$ and $(A-B)(A-B)^\top \mge (1-\eps)AA^\top + \tp{1-\frac{1}{\eps}}BB^\top$.
\end{lemma}
\begin{proof}
	Obviously 
	\[
	\eps AA^\top + \frac{1}{\eps} BB^\top-AB^\top - BA^\top = \tp{\sqrt{\eps}A-\frac{1}{\sqrt{\eps}}B}\tp{\sqrt{\eps}A-\frac{1}{\sqrt{\eps}}B}^\top\mge 0.
	\]
	The second and third inequalities then follow from the first one by expanding the respective LHS.
\end{proof}

We use $\!{supp}\,A$ to denote the support of a matrix $A$, namely the collection of coordinates with nonzero value.

\begin{lemma}\label{lem:matrix-sum-bound}
Let $U$ be a finite set and $\set{U_1,\dots,U_n}$ be a collection of subsets of $U$. Let $A\in \bb R^{U\times U}$ be a matrix. Assume $A=\sum_{i=1}^n A_i$ where each $A_i$ is a matrix satisfying $\!{supp}\,A_i\subseteq U_i\times U_i$. For every $v\in U$, let $T(v)\defeq \set{i\in [n]\cmid v\in U_i}$. For every $i\in [n]$, let $m_i\defeq \max_{v\in U_i} \abs{T(v)}$. Then
\[
	AA^\top\mle \sum_{i\in [n]} m_i\cdot A_iA_i^\top.
\]
\end{lemma}
\begin{proof}
	Let $A(v)$ denote the row in $A$ indexed by $v \in U$. For any $x \in \bb R^U$,
\begin{equation}\label{eqn:mle-square-sum-1}
    x^{\top}A^\top Ax=\|Ax\|^2=\sum_{v\in U} ((Ax)_v)^2=\sum_{v\in U}\inner{A(v)}{x}^2.
\end{equation}
We write $A(v)=\sum_{i=1}^n A_i(v)=\sum_{i \in T(v)} A_i(v)$. Then
\begin{equation}\label{eqn:mle-square-sum-2}
    \!{RHS}\mbox{ of \eqref{eqn:mle-square-sum-1}}=\sum_{v\in U}\inner{\sum_{i \in T(v)}A_i(v)}{x}^2 =\sum_{v\in U} \left(\sum_{i\in T(v)} \inner{A_i(v)}{x}\right)^2
\end{equation}
By Cauchy-Schwarz inequality, 
\[
 \tp{\sum_{i\in T(v)} \inner{A_i(v)}{x}}^2 \leq |T(v)|\sum_{i \in T(v)} \inner{A_i(v)}{x}^2.
\]
Therefore,
\begin{align*}
	\!{RHS}\mbox{ of \eqref{eqn:mle-square-sum-2}} &\leq \sum_{v \in U} |T(v)|\sum_{i \in T(v)} \inner{A_i(v)}{x}^2
    = \sum_{i=1}^n \sum_{v: i\in T(v)} |T(v)| \inner{A_i(v)}{x}^2
    \leq \sum_{i=1}^n m_i\cdot x^\top A_i^\top A_i x.
\end{align*}
\end{proof}

\begin{corollary}
	\label{lem:matrix-squared-sum}
Let $A_1,\dots,A_n$ be a collection of symmetric matrices and $\Pi\mge 0$. Then
\[
	\tp{\sum_{i=1}^n A_i}\Pi \tp{\sum_{i=1}^n A_i} \mle n\sum_{i=1}^n A_i\Pi A_i.
\]	
\end{corollary}
\begin{proof}
	It follows from \Cref{lem:matrix-sum-bound} directly by substituting $A_i$ by $A_i\Pi^{1/2}$ and noting that $m_i\le n$.
\end{proof}
\begin{lemma}\label{lem:matrix-product}
	Let $A$ be a matrix and $B, B'$ be two symmetric matrices such that $B\mle B'$. Then
	\[
		A^\top BA \mle A^\top B'A.
	\]
\end{lemma}
\begin{proof}
	For any $\*x$, we have 
	\[
		\*x^\top \tp{A^\top B' A-A^\top B A}\*x = \*x^\top A^\top (B'-B)A\*x\ge 0.
	\]
\end{proof}

\begin{lemma}\label{lem:matrix-product-2}
	Let $A,B,\Pi$ be matrices where $B$ and $\Pi$ are diagonal. Assume $\Pi\mge 0$. Then
	\[
		A\mle \Pi B \iff \Pi^{-\frac{1}{2}} A\Pi^{-\frac{1}{2}} \mle B.
	\]
\end{lemma}
\begin{proof}
	Since $B, \Pi$ are diagonal, $\Pi B = \Pi^{1/2} B \Pi^{1/2}$. This lemma is then implied by \Cref{lem:matrix-product}.
\end{proof}


\section{Vertex Coloring on Line Graphs}
\label{sec:vertex-coloring}

In this section, we fix a graph $G=(V,E)$ which is the line graph of $\wh G=(\wh V,\wh E)$ with maximum degree $\wh \Delta$. As a result, $V = \wh E$ and the maximum degree $\Delta$ of $G$ is at most $2\wh \Delta-2$. Let $(G,\+L)$ be a $\beta$-extra list-coloring instance with $\beta\ge 2$. After fixing notations in \Cref{sec:notations}, we will construct matrices fulfilling requirements of \Cref{prop:mtd-inductive} in \Cref{sec:base-case} and \Cref{sec:induction}. Then we prove the main theorems in \Cref{sec:main-proof}.

\subsection{Notations}\label{sec:notations}

We fix some notations that will be used throughout the construction. Some of them might have been introduced in \Cref{sec:pre}. Nevertheless, we summarize here for easier reference.

Let $\Lambda\subseteq V$ and $\tau\in [q]^\Lambda$ be a \emph{partial coloring} of $(G,\+L)$. We define
\begin{itemize}
	\item $\!{col}(\tau)\defeq \set{c\in [q] \cmid \exists v,\;vc \in \tau}$ as the set of colors used by $\tau$;
	\item $V_\tau \defeq \set{v\cmid vc\in \+C_{\tau,1}} = V\setminus\Lambda$ as the set of vertices \emph{not} colored by $\tau$;
	\item $\ol V_\tau \defeq V \setminus V_\tau = \Lambda$ as the set of vertices colored by $\tau$.
\end{itemize}

For any $S\subseteq \ol V_\tau$, we use $\tau|_S \defeq \set{vc \cmid vc \in \tau \land v\in S}$ to denote the partial coloring obtained from $\tau$ by restricting on $S$. Let $G_\tau \defeq G[V_\tau]$ be the subgraph induced by $V_\tau$, and $\Delta_\tau(v)$ be the degree of $v\in V_\tau$ on $G_\tau$. 

We define the color list after pinning $\tau$ as $L_v^\tau \defeq \set{c\in L_v \cmid vc \in \+C_{\tau,1}}$ and $\ell_v^\tau = \abs{L_v^\tau}$ for every $v\in V_\tau$. Similarly, for every two distinct vertices $u,v\in V_\tau$, we define $L_{uv}^\tau \defeq L_u^\tau \cap L_v^\tau$ and $\ell_{uv}^\tau = \abs{L_{uv}^\tau}$. Then $(G_\tau, \+L^\tau)$ is the list-coloring instance after pinning $\tau$ where $\+L^\tau=\set{L^\tau_v}_{v\in V_\tau}$. 

For every $i\in \wh V$, we use $V^i\subseteq V$ to denote the set of edges in $\wh G$ incident to $i$. By the definition of a line graph, $G[V^i]$ is a clique in $G$. Let $V_\tau^i\subseteq V^i$ denote the vertices of $V^i$ not in $\tau$, i.e., $V_\tau^i=V^i\cap V_\tau$.
Fix a color $c \in [q]$. We write $V^c = \set{vc \cmid v\in V \land vc\in \+C_1}$ and $V_\tau^{i, c}=\set{vc \cmid v\in V_\tau^i \land vc\in \+C_{\tau,1}}$. We also use $h^{i,c}_\tau\defeq \abs{V^{i,c}_\tau}-1$ to denote the size of $V_\tau^i$ minus one. Note that $G[V_\tau^i]$ is still a clique after pinning $\tau$. Define $\!{Adj}_\tau^{i, c} \in \bb R^{V^c\times V^c}$ as $\!{Adj}_\tau^{i,c}(uc,vc) = \!{Adj}_\tau^{i,c}(vc,uc) = 1$ for every distinct $uc,vc\in V^{i,c}_\tau$ and all other entries $0$, which is obtained by restricting the adjacency matrix of $G[V_\tau^{i}]$ to the vertices that appear in $V_\tau^{i, c}$. Let $\!{Id}_\tau^{i,c}\in \bb R^{V^c\times V^c}$ be the identity matrix restricted on $V_\tau^{i,c}$, and similarly $\!{Id}^c_\tau\in \bb R^{V^c\times V^c}$ be the identity matrix restricted on $V_\tau^c$.

\subsection{Base Case}
\label{sec:base-case}
The base case of \Cref{prop:mtd-inductive} is that
\[
	\Pi_\tau P_\tau - 2\pi_\tau \pi_\tau^\top \mle M_\tau 	
\]
for any $\tau$ with $\!{codim}(\tau) = 2$. Our construction of $M_\tau$ is similar to the one in~\cite{ALOG21}. 
We can explicitly write down $\Pi_\tau P_\tau - 2\pi_\tau \pi_\tau^\top$ when $\!{codim}(\tau)=2$. Let $(G_\tau, L^\tau)$ be the instance of list-coloring after pinning $\tau$ in $G$. If $G_\tau$ is disconnected, then $\Pi_\tau P_\tau -2\pi_\tau \pi_\tau^\top =\*0$. In this case, we let $M_\tau=\*0$ be the all zero matrix. Otherwise, assume $G_\tau=(\set{u,v},\set{\set{u,v}})$. For the sake of brevity, we drop the superscript $\tau$ of $L$ and $\ell$ in this section. 
For every $c_1\in L_u$ and $c_2\in L_v$, we have
\[
\pi_\tau(uc_1) = \frac{1}{2}\cdot\frac{\ell_v-\1[]{c_1\in L_{uv}}}{\ell_u\ell_v-\ell_{uv}},\quad\pi_\tau(vc_2) = \frac{1}{2}\cdot\frac{\ell_u-\1{c_2\in L_{uv}}}{\ell_u\ell_v-\ell_{uv}}.
\]
Let $\*1_u\in\bb R^{\+C_1}$ be the vector with $1$ on positions indexed by $uc$ for every $c\in L_u$ and $0$ otherwise. Define $\*1_v$ similarly. Let $\*1_{L_{uv}}$ be the vector with $1$ on positions indexed by $uc$ or $vc$ where $c\in L_{uv}$ and $0$ otherwise. We can write $\pi_{\tau}$ as
\[
\pi_\tau = \frac{1}{2(\ell_u\ell_v - \ell_{uv})}\tp{\ell_u\*1_v + \ell_v\*1_u-\*1_{L_{uv}}}.
\]
For every $c_1\in L_u$ and $c_2\in L_v$, we have
\[
	\Pi_\tau P_{\tau}(uc_1,vc_2)= \Pi_\tau P_{\tau}(vc_2,uc_1)=
	\begin{cases}
		\frac{1}{2(\ell_u\ell_v-\ell_{uv})}, & \mbox{ if }c_1\ne c_2;\\
		0, & \mbox{ otherwise.}
	\end{cases}
\]
Let $\*1_{u,v}\in \bb R^{\+C_1\times\+C_1}$ be the matrix where $\*1_{u,v}(uc,vc) =  \*1_{u,v}(vc,uc)=1$ if $c\in L_{uv}$ and the other entries are $0$.

As a result, 
\[
	\Pi_\tau P_\tau = \frac{1}{2(\ell_u\ell_v-\ell_{uv})}\tp{\*1_u\*1_v^\top +\*1_v\*1_u^\top -\*1_{u,v}}.
\]
We can therefore express $\Pi_\tau P_{\tau}-2\pi_\tau\pi_\tau^\top$ as
\begin{align*}
	&\phantom{{}={}}\Pi_\tau P_{\tau}-2\pi_\tau\pi_\tau^\top\\
	&=\frac{1}{2(\ell_u\ell_v-\ell_{uv})^2}\tp{\tp{\ell_u\ell_v-\ell_{uv}}\tp{\*1_u\*1_v^\top+\*1_v\*1_u^\top-\*1_{u,v}}-\tp{\ell_u\*1_v+\ell_v\*1_u-\*1_{L_{uv}}}\tp{\ell_u\*1_v^{\top}+\ell_v\*1_u^\top-\*1_{L_{uv}}^\top}}\\
	&\mle \frac{1}{2(\ell_u\ell_v-\ell_{uv})^2}\tp{\tp{\ell_u\ell_v-\ell_{uv}}\tp{\*1_u\*1_v^\top+\*1_v\*1_u^\top-\*1_{u,v}}-\tp{\frac{\tp{\ell_u\*1_v+\ell_v\*1_u}\tp{\ell_u\*1_v^\top+\ell_v\*1_u^\top}}{2}-\*1_{L_{uv}}\*1_{L_{uv}}^\top}},
\end{align*}
where we used \Cref{lem:matrix-basicineq}.

We now claim that 
\[
	\tp{\ell_u\*1_v+\ell_v\*1_u}\tp{\ell_u\*1_v^\top+\ell_v\*1_u^\top} \mge 2(\ell_u\ell_v-\ell_{uv}	)\tp{\*1_u\*1_v^\top+\*1_v\*1_u^\top}.
\]
It follows from the claim that
\begin{align*}
\Pi_\tau P_{\tau}-2\pi_\tau\pi_\tau^\top
	&\mle \frac{\*1_{L_{uv}}\*1_{L_{uv}}^\top - (\ell_u\ell_v-\ell_{uv})\*1_{u,v}}{2(\ell_u\ell_v-\ell_{uv})^2} = \colon \ol M_\tau.
\end{align*}
For every $c\in [q]$, define the matrix $M_\tau^c\in \bb R^{\set{uc,vc}\times \set{uc,vc}}$ as $M_\tau^c=\*0$ if $c\not\in L_{uv}$ and otherwise
\begin{align}
	\label{eqn:base-case-def-ol}
	M_\tau^c 
	&=
	\frac{1}{2(\ell_u\ell_v-\ell_{uv})}
	\begin{bmatrix}
		0 & -1\\
		-1 & 0
	\end{bmatrix} + \frac{1}{(\beta-1)^2}
	\begin{bmatrix}
		\frac{1}{2}\frac{\ell_v-1}{\ell_u\ell_v-\ell_{uv}} & 0\\
		0 & \frac{1}{2}\frac{\ell_u-1}{\ell_u\ell_v-\ell_{uv}}
	\end{bmatrix} \notag \\
	&=
	\frac{1}{2(\ell_u\ell_v-\ell_{uv})}
	\begin{bmatrix}
		0 & -1\\
		-1 & 0
	\end{bmatrix} + \frac{1}{(\beta-1)^2}
	\Pi_\tau \!{Id}_\tau^c
	.
\end{align}
Let $M_\tau$ be the block-diagonal matrix with block  $M_\tau^c$ on the diagonal for each $c\in L_{uv}$.

Let $I_{u,v}$ be matrix that $I_{u,v}(uc,uc)=I_{u,v}(vc,vc)=1$ if $c\in L_{uv}$ and the other entries are $0$. Observing that the $uc$-th row summation of $\*1_{L_{uv}}\*1_{L_{uv}}^\top$ is at most $\ell_{uv}$, we have 
\[\frac{\*1_{L_{uv}}\*1_{L_{uv}}^\top}{2(\ell_u\ell_v-\ell_{uv})^2} \mle \frac{\ell_{uv}}{2(\ell_u\ell_v-\ell_{uv})^2}I_{u,v}\mle \frac{1}{2(\beta-1)(\ell_u\ell_v-\ell_{uv})}I_{u, v} \mle \!{diag}(M_\tau).\]
 Hence $\ol M_\tau\mle M_\tau$.

\bigskip
It remains to verify the claim. Note that since $\*1_v\*1_v^\top\mge 0$,
\begin{align*}
		\tp{\ell_u\*1_v+\ell_v\*1_u}\tp{\ell_u\*1_v^\top+\ell_v\*1_u^\top}
		&=\ell_u\ell_v\tp{\*1_u\*1_v^\top+\*1_v\*1_u^\top}+\ell_u^2 \*1_u\*1_u^\top + \ell_v^2\*1_v\*1_v^\top \\
		&\mge \ell_u\ell_v\tp{\*1_u\*1_v^\top+\*1_v\*1_u^\top} + \ell_u^2 \*1_u\*1_u^\top + (\ell_v-2\ell_{uv}/\ell_u)^2\*1_v\*1_v^\top.
\end{align*}
Since $\ell_u\geq \beta\geq 2$, the coefficient $\ell_v-2\ell_{uv}/\ell_u\geq 0$. Therefore, by \Cref{lem:matrix-basicineq}, we have 
\[
\ell_u^2 \*1_u\*1_u^\top + (\ell_v-2\ell_{uv}/\ell_u)^2\*1_v\*1_v^\top\mge (\ell_u\ell_v-2\ell_{uv})\tp{\*1_u\*1_v^\top+\*1_v\*1_u^\top}. 
\]
And hence,
\[ 
	\tp{\ell_u\*1_v+\ell_v\*1_u}\tp{\ell_u\*1_v^\top+\ell_v\*1_u^\top}\mge 2(\ell_u\ell_v-\ell_{uv})\tp{\*1_u\*1_v^\top+\*1_v\*1_u^\top}.
\]

%
%
%

\subsection{Induction Step}
\label{sec:induction}

The induction step in \Cref{prop:mtd-inductive} is to show that for every $\tau$ with $\!{codim}(\tau) = k>2$ and connected $G_\tau$,

\begin{equation}\label{eqn:induction-main}
	\E{M_{\tau\cup \set{x}}}\mle M_\tau - \frac{k-1}{k-2}M_\tau \Pi_\tau^{-1} M_\tau.
\end{equation}
For every $\tau$ and $c\in [q]$, we will define a matrix $N_\tau^c \in \bb R^{V^c\times V^c}$ and let $M_\tau$ be the block diagonal matrix with block $N_\tau^c$ for every $c\in [q]$. 
It is not hard to see that we only require 
\begin{equation}\label{eqn:induction-main-c}
	\E{N_{\tau\cup \set{x}}^c}\mle N_\tau^c - \frac{k-1}{k-2}N_\tau^c (\Pi_\tau^c)^{-1} N_\tau^c
\end{equation}
to hold for every $c$ and $\tau$ with connected $G_\tau$, where $(\Pi_\tau^c)^{-1}$ is $\Pi_\tau^{-1}$ restricted on $V^c\times V^c$. We now describe our construction of $N^c_\tau$ for a fixed color $c$. We write $N_\tau^c$ into the sum of a diagonal matrix and an off-diagonal matrix, i.e.,
\begin{equation}\label{eqn:N-decompose}
	N_\tau^c = \frac{1}{k-1}(A_\tau^c+\Pi_\tau^c B_\tau^c),
\end{equation}
where $A_\tau^c$ is an off-diagonal matrix and $B_\tau^c$ is a diagonal matrix. For the off-diagonal matrix $A_\tau^c$, we further decompose it into $A_\tau^c=\sum_{i\in \wh V} A_\tau^{i,c}$, where $A^{i,c}_\tau \in \bb R^{V^c\times V^c}$ with $\!{supp}\,A^{i,c}_\tau\subseteq V^{i,c}_\tau \times V^{i,c}_\tau$.

Note that we can extend the notations to those $\tau$ with disconnected $G_\tau$. Following \Cref{prop:mtd-inductive}, for $\tau$ with disconnected $G_\tau$, that is, when $(\+C_\tau,\pi_{\tau,k})$ is the product of $M$ pure weighted simplicial complexes $(\+C^{(1)},\pi^{(1)}), \dots (\+C^{(M)},\pi^{(M)})$ of dimension $n_1,\dots,n_M$ respectively, we define
			\begin{equation} \label{eqn:induction-disconnected}
				M_\tau = \sum_{i\in [M]\colon n_i\ge 2} \frac{n_i(n_i-1)}{k(k-1)}\cdot M_{\tau\cup \eta_{-i}}
			\end{equation}
			where $\eta_{-i} = \eta\setminus \+C^{(i)}_1$ for an arbitrary $\eta\in \+C_{\tau,k}$.
We can write above $M_\tau$ as the block-diagonal matrices with block $N_\tau^c$ for each $c$ and decompose $N_\tau^c=\frac{1}{k-1}(A_\tau^c+\Pi_\tau^c B_\tau^c)$ as in the connected case. Plugging into \eqref{eqn:induction-disconnected}, we obtain
\begin{equation}\label{eqn:disconnected-A-B}
	A^c_\tau = \sum_{i\in [M]\colon n_i\ge 2} \frac{n_i}{k}\cdot A^c_{\tau\cup \eta_{-i}},\quad B^c_\tau = \sum_{i\in [M]\colon n_i\ge 2} B^c_{\tau\cup \eta_{-i}}.
\end{equation}
We can also decompose $A^c_\tau = \sum_{i\in \wh{V}} A_\tau^{i,c}$ for $\tau$ with disconnected $G_\tau$ similarly.

From now on, when $c$ is clear from the context, we will omit the superscript $c$ for matrices. For example, we will write $N_\tau$, $N^i_\tau$, $\Pi_\tau$, $\!{Adj}_\tau^i$, $\!{Id}_\tau^i$, $A_\tau$, $B_\tau$, $\dots$ instead of $N^c_\tau$, $N^{i,c}_\tau$, $\Pi^c_\tau$, $\!{Adj}_\tau^{i,c}$, $\!{Id}_\tau^{i,c}$, $A_\tau^c$, $B_\tau^c$, $\dots$ respectively. Also we write $h_\tau^i$ for $h_\tau^{i,c}$. Plugging the above construction of $N_\tau$ into \eqref{eqn:induction-main-c} and remembering that the superscript $c$ has been omitted, we obtain
 \[
 (k-1)\cdot\E[x\sim\pi_\tau]{A_{\tau\cup\set{x}}+\Pi_{\tau\cup\set{x}}B_{\tau\cup\set{x}}} \mle (k-2)\cdot\tp{A_\tau+\Pi_\tau B_\tau} -\tp{A_\tau+\Pi_\tau B_\tau}\Pi_\tau^{-1}\tp{A_\tau+\Pi_\tau B_\tau}.
 \]
 It follows from \Cref{lem:matrix-squared-sum}  that
 \[
 \tp{A_\tau+\Pi_\tau B_\tau}\Pi_\tau^{-1}\tp{A_\tau+\Pi_\tau B_\tau}\mle 2A_\tau\Pi_\tau^{-1}A_\tau + 2\Pi_\tau B_\tau^2.
 \]
 Since each vertex $v$ only occurs in $V^i$ and $V^j$ assuming $v=\set{i,j}$ in $\wh G$, by \Cref{lem:matrix-sum-bound},
 \[
 A_\tau\Pi_\tau^{-1}A_\tau\mle 2\sum_{i\in \wh V} A_\tau^i \Pi_\tau^{-1} A_\tau^i.
 \]
 As a result, in order for \eqref{eqn:induction-main} to hold, we only need to design $A_\tau^i$ for every $i\in \wh V$ and $B_\tau$  satisfying
\begin{equation}\label{eqn:condition-main}
 	\sum_{i\in \wh V} \tp{(k-1)\cdot \E[x\sim\pi_\tau]{A_{\tau\cup\set{x}}^i} - (k-2) \cdot A_\tau^i+4A_\tau^i \Pi_\tau^{-1} A_\tau^i}  \mle (k-2)\Pi_\tau B_\tau -(k-1)\cdot \E[x\sim\pi_\tau]{\Pi_{\tau\cup\set{x}} B_{\tau\cup\set{x}}}-2\Pi_\tau\tp{B_\tau}^2.
\end{equation}

\subsubsection{Construction of $A_\tau^i$}
A natural starting point (as did in~\cite{ALOG21}) is to recursively define  $A_\tau^i = \frac{k-1}{k-2}\cdot \E[x\sim\pi_\tau]{A_{\tau\cup\set{x}}^i}$ and this yields an explicit expression for $A_\tau^i$. However, under this construction, the LHS of~\eqref{eqn:condition-main} becomes $4A_\tau^i \Pi_\tau^{-1} A_\tau^i$, which is too large for ~\eqref{eqn:condition-main} to be feasible in the regime we are interested in. Nevertheless, it is still helpful to see what $A_\tau^i$ looks like under this recursive definition.

For the base case $V_\tau =\set{u,v}$, if $c \in L_u^\tau \cap L_v^\tau$ and $i \in \wh V$ is the common end vertex of $u, v$ in $\wh G$, let $A_\tau^i=\frac{1}{2(\ell_u^\tau\ell_v^\tau-\ell_{uv}^\tau)}
	\begin{bmatrix}
		0 & -1\\
		-1 & 0
	\end{bmatrix}$. Otherwise set $A_\tau^i=0$. It is obvious that $\sum_{i\in \wh V} A_\tau^i = \!{offdiag}(M_\tau^c)$.
	
When $\!{codim}\,\tau>2$, we can expand the recursion down to faces of dimension $2$:
\begin{align*}
	A_\tau^i 
	&= \frac{k-1}{k-2}\cdot \frac{k-2}{k-3}\cdots \frac{3}{2}\cdot\frac{2}{1} \cdot \E[x_1\sim \pi_\tau]{\E[x_2\sim \pi_{\tau\cup \set{x_1}}]{\dots \E[x_{k-2}\sim \pi_{\tau\cup\set{x_1,\dots,x_{k-3}}}]{A_{\tau\cup\set{x_1,\dots,x_{k-2}}}^i}}}\\
	&=(k-1)\sum_{(x_1,x_2,\dots,x_{k-2}) \in \+C_{\tau,k-2}}  \pi_\tau(x_1)\pi_{\tau\cup \set{x_1}}(x_2)\cdots \pi_{\tau\cup\set{x_1,\dots,x_{k-3}}}(x_{k-2}) A_{\tau \cup \set{x_1,x_2,\dots,x_{k-2}}}^i\\
	&=(k-1)\sum_{(x_1,x_2,\dots,x_{k-2}) \in \+C_{\tau,k-2}} \frac{\pi_{\tau,k-2}(\set{x_1,x_2,\dots,x_{k-2}})}{(k-2)!}A_{\tau \cup \set{x_1,x_2,\dots,x_{k-2}}}^i\\
	&=(k-1)\cdot\E[\sigma\sim\pi_{\tau, k-2}]{A_{\tau\cup\sigma}^i}.
\end{align*}

Recall that $h_\tau^i=\abs{V_\tau^{i,c}}-1$. In order to decrease the LHS of \cref{eqn:condition-main}, we introduce a collection of positive coefficients $a_h$ decreasing in $h$ for $1\le h\le \Delta$ whose value will be determined later. Especially $a_0=0$. For connected $\tau$, define
\begin{equation}\label{eqn:A-def}
	A_\tau^i = a_{h_\tau^i}\cdot (k-1)\cdot \E[\sigma\sim \pi_{\tau, k-2}]{A_{\tau\cup\sigma}^i}.
\end{equation}
When $h_\tau^i=0$, $A_\tau^i=\*0$ is trivial. So in the following analysis, we assume $h_\tau^i \geq 1$.

Note that the above relation holds for $\tau$ with disconnected $G_\tau$.
\begin{lemma}\label{lem:disconnected-tau}
	 For $\tau$ with disconnected $G_\tau$, the identity \eqref{eqn:A-def} holds.
\end{lemma}
\begin{proof}
	Fix $i\in \wh{V}$ and $uc, vc \in V^{i,c}$. We assume the connected component of $G_\tau$ containing $u, v$ is indexed by $j$. Then we have
	\begin{align*}
	A_\tau^i(uc,vc)
	&=\frac{n_j}{k}\cdot A_{\tau\cup \eta_{-j}}^i(uc,vc)\\
	&=\frac{n_j}{k}\cdot a_{h_{\tau\cup \eta_{-j}}^i}\cdot (n_j-1) \cdot \E[\sigma\sim \pi_{\tau \cup \eta_{-j},n_j-2}]{ A_{\tau\cup\eta_{-j}\cup\sigma}^i(uc,vc)}\\
	&=\frac{2}{k} \cdot a_{h_{\tau}^i} \cdot \sum_{\sigma\in \+C_{\tau\cup \eta_{-j},n_j-2}} \mu^{\tau\cup \eta_{-j}}_{V_{\tau\cup \eta_{-j}}\setminus \set{u,v}}(\sigma)\cdot A_{\tau\cup\eta_{-j}\cup\sigma}^i(uc,vc).
	\end{align*}
	Since conditioning on $\eta_{-j}$ does not affect the distribution of $\sigma$, we can further write above as
	\begin{align*}
	\frac{2}{k} \cdot a_{h_{\tau}^i} \cdot \sum_{\sigma\in \+C_{\tau,k-2}} \mu^{\tau}_{V_{\tau}\setminus \set{u,v}}(\sigma)A_{\tau\cup\sigma}^i(uc,vc)
	=a_{h_{\tau}^i} \cdot (k-1) \cdot \sum_{\sigma\in \+C_{\tau,k-2}} \pi_{\tau,k-2}(\sigma)A_{\tau\cup\sigma}^i(uc,vc).
	\end{align*}
\end{proof}
As a result, \eqref{eqn:A-def} holds for all $\tau$, and we can deduce the following relation between $A_\tau^i$'s whose co-dimensions differ by one.

\begin{lemma}
	\[
		\E[x\sim \pi_\tau]{A_{\tau\cup\set{x}}^i} = \frac{1}{k-1}\tp{(h-1)\frac{a_{h-1}}{a_h} + (k-2-(h-1))}\cdot A_\tau^i,
	\]
	where $h = h_\tau^i \geq 1$.
\end{lemma}
\begin{proof}
For any $uc, vc \in V_\tau^{i,c}$,
	\begin{align*}
		\E[x\sim \pi_\tau]{A_{\tau\cup\set{x}}^i}(uc, vc) &= (k-2)\sum_{x\in \+C_{\tau,1}} \frac{1}{k}\mu^\tau_{\ol V_x}(x) a_{h_{\tau\cup\set{x}}^i} \sum_{\sigma\in \+C_{\tau\cup\set{x}, k-3}} \frac{2}{(k-1)(k-2)}\mu^{\tau\cup\set{x}}_{\ol V_\sigma}(\sigma) A_{\tau\cup\set{x}\cup\sigma}^i(uc, vc)\\
		&= \frac{2}{k(k-1)} \sum_{x\in \+C_{\tau,1}}\sum_{\sigma\in \+C_{\tau\cup\set{x}, k-3}}\mu^{\tau}_{\ol V_x}(x) \mu^{\tau\cup\set{x}}_{\ol V_\sigma}(\sigma)a_{h_{\tau\cup\set{x}}^i}A_{\tau\cup\set{x}\cup\sigma}^i(uc, vc)\\
		&= \frac{2}{k(k-1)} \sum_{\sigma'\in \+C_{\tau,k-2}} \mu^\tau_{\ol V_{\sigma'}}(\sigma') A_{\tau\cup\sigma'}^i(uc, vc) \sum_{x \in \sigma'} a_{h_{\tau\cup\set{x}}^i}\\
		&= \frac{2}{k(k-1)} \sum_{\sigma'\in \+C_{\tau,k-2}} \mu^\tau_{\ol V_{\sigma'}}(\sigma') A_{\tau\cup\sigma'}^i(uc, vc)\tp{(h-1)a_{h-1} + (k-2-(h-1))a_h}\\
		&= \frac{1}{k-1} A_\tau^i(uc, vc)\tp{(h-1)\frac{a_{h-1}}{a_h} + (k-2-(h-1))}.
	\end{align*}
\end{proof}

We remark that \Cref{lem:disconnected-tau} is essential in the above proof since pinning a single $x$ might result in disconnected $G_{\tau\cup\set{x}}$.


It follows from the definition that $A_\tau^i$ is proportional to the expectation of the base cases when the boundary is drawn from $\pi_{\tau,k-2}$. For some technical reasons, we would like to isolate those boundaries containing the color $c$. This leads us to the following lemma.

\begin{lemma} \label{lem:A-tau-i}
	\[
		A_\tau^i = \frac{a_h}{k}\cdot\frac{1}{\abs{\+C_{\tau,k}}}\cdot \sum_{\substack{\omega\in \+C_{\tau,k}\colon\\\forall w\in V^i_\tau \colon \omega(w)\ne c}} A^{i,\omega}_\tau,
	\]
	where $A^{i,\omega}_\tau$ is the matrix satisfying for any $uc, vc \in V_\tau^{i,c}$,
	\[
		A_\tau^{i,\omega}(uc,vc) = -\frac{\mathbbm{1}\left[c\in L^{(\tau\cup \omega)|_{V\setminus\set{u,v}}}_u\right]\cdot \mathbbm{1}\left[c\in L^{(\tau\cup \omega)|_{V\setminus\set{u,v}}}_v\right]}{\tp{\ell^{(\tau\cup \omega)|_{V\setminus\set{u,v}}}_u-1}\cdot\tp{\ell^{(\tau\cup \omega)|_{V\setminus\set{u,v}}}_v-1} - \tp{\ell^{(\tau\cup \omega)|_{V\setminus\set{u,v}}}_{uv}-1}}
	\]
	and all other entries $0$.
\end{lemma}
\begin{proof}
Let $V'\defeq V_\tau\setminus \set{u,v}$. We denote the set of proper partial colorings restricted on $V' \subset V_\tau$ when $\tau$ is pinned by $\+C_{\tau}^{V'}$. Formally, $\+C_{\tau}^{V'}=\set{\omega \in \+C_{\tau} \cmid \ol V_{\omega}=V'}$. It follows from \cref{eqn:A-def} that for any $uc, vc \in V_\tau^{i,c}$,
\begin{align}
	A_\tau^i(uc,vc)
&=a_h\cdot (k-1) \sum_{\sigma\in \+C_{\tau,k-2}} \pi_\tau(\sigma)\cdot A^i_{\tau\cup\sigma}(uc,vc)\notag\\
&=-\frac{a_h}{k}\sum_{\sigma\in \+C_\tau^{V'}}\mu_{V'}^\tau(\sigma)\cdot \frac{\mathbbm{1}\left[c\in L^{\tau\cup\sigma}_u\right]\cdot \mathbbm{1}\left[c\in L^{\tau\cup\sigma}_v\right]}{\ell^{\tau\cup \sigma}_u\cdot\ell^{\tau\cup\sigma}_v - \ell^{\tau\cup\sigma}_{uv}}\notag\\
&=-\frac{a_h}{k}\cdot \sum_{\sigma\in \+C_\tau^{V'}}\frac{\abs{\set{\omega\in\+C_{\tau,k}\cmid \sigma\subseteq \omega}}}{\abs{\+C_{\tau,k}}} \cdot \frac{\mathbbm{1}\left[c\in L^{\tau\cup\sigma}_u\right]\cdot \mathbbm{1}\left[c\in L^{\tau\cup\sigma}_v\right]}{\ell^{\tau\cup \sigma}_u\cdot\ell^{\tau\cup\sigma}_v - \ell^{\tau\cup\sigma}_{uv}}.
\label{eqn:A1}
\end{align}
Note that for every $\sigma\in\+C_\tau^{V'}$, we have 
\[
\abs{\set{\omega\in\+C_{\tau,k}\cmid \sigma\subseteq \omega}} = \abs{\+C_{\tau\cup\sigma,2}} = \ell^{\tau\cup \sigma}_u\cdot\ell^{\tau\cup\sigma}_v - \ell^{\tau\cup\sigma}_{uv}.
\]
Plugging this into \cref{eqn:A1}, we have
\[
	A_\tau^i(uc,vc) = -\frac{a_h}{k}\cdot \frac{1}{\abs{\+C_{\tau,k}}}\cdot \sum_{\sigma\in\+C_\tau^{V'}} \mathbbm{1}\left[c\in L^{\tau\cup\sigma}_u\right]\cdot \mathbbm{1}\left[c\in L^{\tau\cup\sigma}_v\right].
\]
Observe that if for a $\sigma\in\+C_\tau^{V'}$, it holds that $\mathbbm{1}\left[c\in L^{\tau\cup\sigma}_u\right]\cdot \mathbbm{1}\left[c\in L^{\tau\cup\sigma}_v\right] = 1$, then 
\[
	\abs{\set{\rho \in\+C_{\tau\cup\sigma,2}\cmid c\not\in\!{col}(\rho)}} = \tp{\ell^{\tau\cup \sigma}_u-1}\cdot\tp{\ell^{\tau\cup\sigma}_v-1} - \tp{\ell^{\tau\cup\sigma}_{uv}-1}.
\]
We can further write $A_\tau^i(uc,vc)$ as 
\begin{align*}
	A_\tau^i(uc,vc)
	&= -\frac{a_h}{k}\cdot \frac{1}{\abs{\+C_{\tau,k}}}\cdot \sum_{\sigma\in\+C_\tau^{V'}} \abs{\set{\rho \in\+C_{\tau\cup\sigma, 2}\cmid c\not\in\!{col}(\rho)}}\cdot \frac{\mathbbm{1}\left[c\in L^{\tau\cup\sigma}_u\right]\cdot \mathbbm{1}\left[c\in L^{\tau\cup\sigma}_v\right]}{\tp{\ell^{\tau\cup \sigma}_u-1}\cdot\tp{\ell^{\tau\cup\sigma}_v-1} - \tp{\ell^{\tau\cup\sigma}_{uv}-1}}\\
	&=-\frac{a_h}{k}\cdot\frac{1}{\abs{\+C_{\tau,k}}}\cdot \sum_{\substack{\omega\in \+C_{\tau,k}\colon\\\omega(u)\ne c\land \omega(v)\ne c}}  \frac{\mathbbm{1}\left[c\in L^{(\tau\cup \omega)|_{V\setminus\set{u,v}}}_u\right]\cdot \mathbbm{1}\left[c\in L^{(\tau\cup \omega)|_{V\setminus\set{u,v}}}_v\right]}{\tp{\ell^{(\tau\cup \omega)|_{V\setminus\set{u,v}}}_u-1}\cdot\tp{\ell^{(\tau\cup \omega)|_{V\setminus\set{u,v}}}_v-1} - \tp{\ell^{(\tau\cup \omega)|_{V\setminus\set{u,v}}}_{uv}-1}}\\
	&=\frac{a_h}{k}\cdot\frac{1}{\abs{\+C_{\tau,k}}}\cdot \sum_{\substack{\omega\in \+C_{\tau,k}\colon\\\omega(u)\ne c\land \omega(v)\ne c}} A^{i,\omega}_\tau(uc,vc)\\
	&=\frac{a_h}{k}\cdot\frac{1}{\abs{\+C_{\tau,k}}}\cdot \sum_{\substack{\omega\in \+C_{\tau,k}\colon\\\forall w\in V^i_\tau\colon \omega(w)\ne c}} A^{i,\omega}_\tau(uc,vc),
\end{align*}
where the last line follows from the fact that $c\in L^{(\tau\cup \omega)|_{V\setminus\set{u,v}}}_u \cap L^{(\tau\cup \omega)|_{V\setminus\set{u,v}}}_v$ implies $\forall w\in V_\tau^i \setminus \set{u,v}, \omega(w)\neq c$.
\end{proof}


\subsubsection{Spectral analysis of $A_\tau^i$}

In the following lemma, we show that each matrix $A^{i,\omega}_\tau$ can be written as the sum of two matrices which we call the \emph{main term} and the \emph{remainder} respectively. The main term only depends on the adjacency matrix $Adj_\tau^i$ and $L_u$ terms under various boundary conditions and is irrelevant to $L_{uv}$ terms. All the effects of $L_{uv}$ terms are collected in the remainder.

For every $\omega\in\+C_{\tau,k}$ such that $c\not\in \!{col}((\omega\cup\tau)|_{V^i_\tau})$, define $\Xi_\tau^{i,\omega} \in \bb R^{V^c\times V^c}$ as the diagonal matrix such that for every $uc \in V_\tau^{i,c}$:
\[
	\Xi_\tau^{i,\omega}(uc,uc) = \frac{\mathbbm{1}\left[c\in L^{(\tau\cup \omega)|_{V\setminus\set{u}}}_u\right]}{\ell^{(\tau\cup \omega)|_{V\setminus\set{u}}}_u-1}.
\]
\begin{lemma}\label{lem:A-tau-i-omega}
	\[
	A^{i,\omega}_\tau = \Xi^{i,\omega}_\tau (-\!{Adj}_\tau^i + \+R_\tau^{i,\omega}) \Xi_{\tau}^{i,\omega} ,
	\]
	where $\rho(\+R_\tau^{i,\omega}) \le  \frac{2 h}{\beta - 1}$.
\end{lemma}
\begin{proof}
	Let $uc,vc\in V_\tau^{i,c}$. To ease the notation, when $\tau$ and $\omega$ are clear from the context, we use $\Gamma(u)$ and $\Gamma(u,v)$ to denote the partial coloring $(\tau\cup\omega)|_{V\setminus\set{u}}$ and $(\tau\cup\omega)|_{V\setminus\set{u,v}}$ respectively. We also use $\ol\ell_u^{\Gamma(u,v)}$ to denote $\ell_u^{\Gamma(u,v)}-1$ and define $\ol\ell_v^{\Gamma(u,v)}, \ol\ell_u^{\Gamma(u)}, \ol\ell_v^{\Gamma(v)}$ similarly.

	Using our new notations, we have
	\[
	\Xi_\tau^i(uc,uc) = \frac{\mathbbm{1}\left[c\in L^{\Gamma(u)}_u\right]}{\ol\ell^{\Gamma(u)}_u}.
	\]
		
	Observing that since $c\not\in\!{col}\tp{(\tau\cup \omega)|_{V^i_\tau}}$, we have
	\begin{align*}
		\mathbbm{1}\left[c\in L^{\Gamma(u)}_u\right] &= \mathbbm{1}\left[c\in L^{\Gamma(u,v)}_u\right]\\
		\mathbbm{1}\left[c\in L^{\Gamma(v)}_v\right] &= \mathbbm{1}\left[c\in L^{\Gamma(u,v)}_v\right].
	\end{align*}
	So we can write $A_\tau^{i,\omega}$ as

	\begin{align*}
		A_\tau^{i,\omega}(uc,vc)
		&=-\frac{\mathbbm{1}\left[c\in L^{\Gamma(u)}_u\right]\cdot \mathbbm{1}\left[c\in L^{\Gamma(v)}_v\right]}{\ol\ell^{\Gamma(u,v)}_u \ol\ell^{\Gamma(u,v)}_v - \ol\ell^{\Gamma(u,v)}_{uv}} \\
		&=\mathbbm{1}\left[c\in L^{\Gamma(u)}_u\right] \cdot \mathbbm{1}\left[c\in L^{\Gamma(v)}_v\right] \cdot \frac{1}{\ol\ell_u^{\Gamma(u)}\cdot \ol\ell_v^{\Gamma(v)}}\tp{-1 + \+R_\tau^{i,\omega}(uc,vc)},
	\end{align*}
	where 
	\begin{equation}\label{eqn:eta-remain}
		\+R_\tau^{i,\omega}(uc,vc) =1 - \frac{\ol\ell_u^{\Gamma(u)}\cdot \ol\ell_v^{\Gamma(v)}}{\ol\ell^{\Gamma(u,v)}_u \ol\ell^{\Gamma(u,v)}_v - \ol\ell^{\Gamma(u,v)}_{uv}}.
	\end{equation}
By definition, it holds that $\ell^{\Gamma(u,v)}_u-\ell_u^{\Gamma(u)} \in \set{0,1}$ and $\ell^{\Gamma(u,v)}_v-\ell_v^{\Gamma(v)} \in \set{0,1}$. So we have $\ol\ell^{\Gamma(u,v)}_u \ol\ell^{\Gamma(u,v)}_v-\ol\ell_u^{\Gamma(u)}\cdot \ol\ell_v^{\Gamma(v)} \in \set{0, \ol\ell^{\Gamma(u,v)}_u,\ol\ell^{\Gamma(u,v)}_v ,\ol\ell^{\Gamma(u,v)}_u+\ol\ell^{\Gamma(u,v)}_v -1}$. Therefore,
	\begin{align*}
		\abs{\+R_\tau^{i,\omega}(uc,vc)} &= \frac{\abs{\ol\ell_u^{\Gamma(u)}\cdot \ol\ell_v^{\Gamma(v)}-\tp{\ol\ell^{\Gamma(u,v)}_u \ol\ell^{\Gamma(u,v)}_v - \ol\ell^{\Gamma(u,v)}_{uv}}}}{\ol\ell^{\Gamma(u,v)}_u \ol\ell^{\Gamma(u,v)}_v - \ol\ell^{\Gamma(u,v)}_{uv}}\\
	&\leq \frac{\max\set{\ol\ell_u^{\Gamma(u,v)}+ \ol\ell_v^{\Gamma(u,v)}, \ol\ell^{\Gamma(u,v)}_{uv}}}{\ol\ell^{\Gamma(u,v)}_u \ol\ell^{\Gamma(u,v)}_v - \ol\ell^{\Gamma(u,v)}_{uv}}\\
	&\leq \frac{\ol\ell^{\Gamma(u, v)}_{u}+ \ol\ell^{\Gamma(u,v)}_{v}}{\ol\ell^{\Gamma(u,v)}_u \ol\ell^{\Gamma(u,v)}_v - \max\set{\ol\ell^{\Gamma(u,v)}_{u},\ol\ell^{\Gamma(u,v)}_{v}}}\\
	&\leq \frac{2}{\min\set{\ol\ell^{\Gamma(u,v)}_{u}, \ol\ell^{\Gamma(u,v)}_{v}}-1}\\
	&\leq \frac{2}{\beta-1},
	\end{align*}
	where the last inequality follows from the fact that $\ol\ell^{\Gamma(u, v)}_{u},\ol\ell^{\Gamma(u, v)}_{v}\geq \beta$.
	Taking row summation of $\+R_\tau^{i,\omega}$, we obtain that $\rho(\+R_\tau^{i,\omega}) \le  \frac{2h}{\beta-1}$.
\end{proof}


In order to bound LHS of \Cref{eqn:condition-main}, we introduce the following lemma.
\begin{lemma}[\cite{GKM15}] \label{lem:marginal-bound}
	For every $(u,c)\in \+C_{\tau,1}$,  it holds
	\[
	\frac{(1-\beta^{-1})^{\Delta_\tau(u)}}{\ell^\tau_u}	\le \mu^{\tau}_u(c) \le \frac{1}{\ell^\tau_u - \Delta_\tau(u)}.
	\]
\end{lemma}
The proof of \Cref{lem:marginal-bound} is included in \Cref{sec:marginal-bound}. 
We assume $\beta \geq \frac{\Delta}{\iota(\Delta)}+1$ where $\iota(\Delta)> 0$ is a slowly increasing function. In the following two lemmas, we bound $A_\tau^{i,\omega}\Pi_\tau^{-1}A_\tau^{i,\omega}$. 

Let $\+C'_{\tau,k}\defeq \set{\omega\in \+C_{\tau,k}\cmid \forall w\in V^i_\tau\colon \omega(w)\ne c}$ be the set of partial colorings in $\+C_{\tau,k}$ where none of vertices in $V^i_\tau$ is colored $c$.

\begin{lemma}\label{lem:bound-A-tau-i}
 $A_\tau^i\Pi_\tau^{-1}A_\tau^i \mle \frac{a_h^2\abs{\+C'_{\tau,k}}}{k^2\abs{\+C_{\tau,k}}^2} \sum_{\omega\in\+C'_{\tau,k}} A_\tau^{i,\omega}\Pi^{-1}_\tau A_\tau^{i,\omega}$.
\end{lemma}
\begin{proof}
	It follows from \Cref{lem:A-tau-i} that 
\begin{equation}\label{eqn:A-tau-i-1}
	A_\tau^i\Pi_\tau^{-1}A_\tau^i
	=\tp{\frac{a_h}{k\cdot \abs{\+C_{\tau,k}}}}^2\tp{\sum_{\omega \in \+C'_{\tau,k}} A_\tau^{i,\omega}}\Pi^{-1}_\tau\tp{\sum_{\omega \in \+C'_{\tau,k}} A_\tau^{i,\omega}}
	\stackrel{(\heartsuit)}{\mle} \frac{a_h^2\abs{\+C'_{\tau,k}}}{k^2\abs{\+C_{\tau,k}}^2} \sum_{\omega\in\+C'_{\tau,k}} A_\tau^{i,\omega}\Pi^{-1}_\tau A_\tau^{i,\omega},
\end{equation}
where $(\heartsuit)$ is due to \Cref{lem:matrix-squared-sum}.
\end{proof}
In the following discussion, let $\gamma=\frac{2(
\beta/(\beta-1)+\iota(\Delta))\exp\tp{\iota(\Delta)}}{\beta-1}$.
\begin{lemma}\label{lem:bound-A-tau-i-omega-square}
	$ A_\tau^{i,\omega} \Pi_\tau^{-1} A_\tau^{i,\omega} \mle \gamma k\cdot \Xi_\tau^{i,\omega}\tp{\tp{\!{Adj}_\tau^i}^2  + \frac{4h^2}{(\beta-1)^2}\!{Id}_\tau^i}\Xi_\tau^{i,\omega}$.
\end{lemma}
\begin{proof}
	By \Cref{lem:marginal-bound}, for any $uc \in V_\tau^{i,c}$,
	\begin{equation}
			\Xi_\tau^{i,\omega}\Pi_\tau^{-1}(uc, uc)\leq \frac{\ell_u^\tau}{\ell_u^\tau - \Delta_\tau(u) - 1} k \exp\tp{\frac{\Delta_\tau(u)}{\beta}} \leq k \tp{\frac{\beta}{\beta-1}+\iota(\Delta)}\exp\tp{\iota(\Delta)}.
	\end{equation}
	Then it follows from \Cref{lem:A-tau-i-omega} and $\Xi_\tau^{i,\omega}\mle \frac{1}{\beta-1}\cdot \!{Id}_{\tau}^i$ that
	\begin{align*}
	 A_\tau^{i,\omega} \Pi_\tau^{-1} A_\tau^{i,\omega} \mle k\gamma \Xi_\tau^{i,\omega}\tp{\tp{\!{Adj}_\tau^i}^2  + \frac{4h^2}{(\beta-1)^2}\!{Id}_\tau^i}\Xi_\tau^{i,\omega}.
	\end{align*}

\end{proof}

\begin{lemma}\label{lem:bound-Xi-Pi}
	$\frac{1}{k\abs{\+C_{\tau,k}}} \Pi_\tau^{-1}\sum_{\omega\in \+C_{\tau,k}'} \Xi_\tau^{i,\omega} = \!{Id}_\tau^i$.
\end{lemma}
\begin{proof}
Let $uc \in V_\tau^{i,c}$. Recall that $\pi_\tau(uc)=\frac{\abs{\set{\sigma\in \+C_{\tau,k}\cmid \sigma(u)=c}}}{k\abs{\+C_{\tau,k}}}$. As we did in \Cref{lem:A-tau-i}, we use the notation $C_\tau^{V'}$ to represent the set of proper partial colorings restricted on $V' \subset V_\tau$ when $\tau$ is pinned.  Therefore, 
	\begin{align*}
		\frac{1}{k\abs{\+C_{\tau,k}}} \pi_\tau^{-1}(uc)\sum_{\omega\in \+C_{\tau,k}'} \Xi_\tau^{i,\omega}(uc,uc) 
		&= \frac{1}{\abs{\set{\sigma\in \+C_{\tau,k}\cmid \sigma(u)=c}}} \sum_{\omega\in \+C'_{\tau,k}} \frac{1}{\abs{\set{\sigma\in \+C_{(\tau \cup \omega)|_{V\setminus \set{u}},1}\cmid \sigma(u)\neq c}}}\\
		& = \frac{1}{\abs{\set{\sigma\in \+C_{\tau,k}\cmid \sigma(u)=c}}} \sum_{\omega'\in \+C'^{V_\tau\setminus \set{u}}_\tau} \frac{\abs{\set{\sigma\in \+C_{\tau \cup \omega', 1}\cmid\sigma(u)\neq c}}}{\abs{\set{\sigma\in \+C_{\tau \cup \omega', 1}\cmid \sigma(u)\neq c}}}\\
		&= \frac{\abs{\+C'^{V_\tau\setminus \set{u}}_\tau}}{\abs{\set{\sigma\in \+C_{\tau,k}\cmid \sigma(u)=c}}}\\
		&= 1,
	\end{align*}
where the second equation is obtained by taking the summation over the color of $u$ when others are fixed.
\end{proof}
We are now ready to bound $\tp{(k-1)\cdot \E[x\sim\pi_\tau]{A_{\tau\cup\set{x}}^i} - (k-2) \cdot A_\tau^i+4A_\tau^i \Pi_\tau^{-1} A_\tau^i}$ in the LHS of \Cref{eqn:condition-main}.

\begin{lemma}
	\label{lem:bound-A-tau}
	There exists a sequence of non-negative numbers $\set{a_h}_{0\le h\le \Delta}$ such that
	\[
		\Pi_\tau^{-\frac{1}{2}} \tp{(k-1)\cdot \E[x\sim\pi_\tau]{A_{\tau\cup\set{x}}^i} - (k-2) \cdot A_\tau^i+4A_\tau^i \Pi_\tau^{-1} A_\tau^i} \Pi_\tau^{-\frac{1}{2}} \mle \frac{8 h (\beta/(\beta-1)+\iota(\Delta))\exp\tp{\iota(\Delta)}}{(\beta-1)^2}  \tp{\frac{2(\beta+1)}{\beta-1}\iota(\Delta)+1}\!{Id}_\tau^i.
	\]
\end{lemma}
\begin{proof}
	Applying \Cref{lem:A-tau-i} and \Cref{lem:bound-A-tau-i}, we obtain
\[
	\!{LHS} \mle \frac{1}{k\abs{\+C_{\tau,k}}} \Pi_\tau^{-\frac{1}{2}} \tp{(h-1)(a_{h-1}-a_h) \sum_{\omega\in\+C'_{\tau,k}} A_\tau^{i,\omega} + \frac{4 a_h^2}{k}\cdot \frac{\abs{\+C'_{\tau,k}}}{\abs{\+C_{\tau,k}}}\sum_{\omega\in \+C'_{\tau,k}}A_\tau^{i,\omega}\Pi^{-1}_\tau A_\tau^{i,\omega}}\Pi_\tau^{-\frac{1}{2}}.
\]
Then by \Cref{lem:A-tau-i-omega} and \Cref{lem:bound-A-tau-i-omega-square}, we can bound above by
\begin{equation}
	\label{eq:lem-bound-A-tau-1}
	\frac{1}{k\abs{\+C_{\tau,k}}}\Pi_\tau^{-\frac{1}{2}}\tp{\sum_{\omega\in\+C'_{\tau,k}}\Xi_\tau^{i,\omega}\tp{-(h-1)(a_{h-1}-a_h)\!{Adj}_\tau^i + 4a_h^2\gamma\tp{\!{Adj}_\tau^i}^2}\Xi_\tau^{i,\omega} + \!{Remainder}}\Pi_\tau^{-\frac{1}{2}},
\end{equation}
where $\gamma=\frac{2(\beta/(\beta-1)+\iota(\Delta))\exp\tp{\iota(\Delta)}}{\beta-1}$ and $\!{Remainder} = \tp{(h-1)(a_{h-1}-a_h) \frac{2h}{\beta-1}+ 4a_h^2\gamma \frac{4h^2}{(\beta-1)^2}}\sum_{\omega\in\+C'_{\tau,k}}\tp{\Xi_\tau^{i,\omega}}^2$.

Naturally, we want to find a sequence of $\set{a_h}$ so that the spectral radius of the following matrices $\tilde A_h$ appearing in the non-remainder terms in \Cref{eq:lem-bound-A-tau-1} is small:
\[
	\tilde A_h\defeq -(h-1)(a_{h-1}-a_h)\!{Adj}_\tau^i + 4a_h^2\gamma\tp{\!{Adj}_\tau^i}^2.
\]
Since the spectrum of $\!{Adj}_\tau^i$ is $\set{-1, h}$, the spectrum of $\tilde A_h$ is
\[
	\set{
		(h-1)(a_{h-1}-a_h) + 4\gamma a_h^2, \;
		-h(h-1)(a_{h-1}-a_h) + 4\gamma h^2 a_h^2
	}.
\]
We want $\rho\tuple{\tilde A_h}$ to be of order $o(h^2)$. This can be achieved by picking $a_h$ as a solution to the recurrence relation $-(a_{h-1} - a_h) + 4\gamma a_h^2 = o(h^2), a_1=1$. The solution we choose is
\[
	a_h = \frac{1}{1+4\gamma(h-1)} (1\leq h \leq \Delta).
\]
Then we have 
\begin{equation}\label{eqn:spectral-radius-of-A}
	\rho(\tilde A_h)\le \frac{4 \gamma  (1+4\gamma(h-1)) h}{(4 \gamma  (h-2)+1) (4 \gamma  (h-1)+1)^2}\le 4\gamma h
\end{equation}
for $h \geq 2$. In particular, when $h=1$, $\rho(\tilde A_h)\leq 4 a_h^2 \gamma h^2 = 4\gamma$, which is consistent with the above bound. So we have $\rho(\tilde A_h)\leq 4\gamma h$ for $h\geq 1$.
Note that $\Xi_\tau^{i,\omega}\mle \frac{1}{\beta-1}\cdot \!{Id}_{\tau}^i$, it then follows from \Cref{eqn:spectral-radius-of-A} and \Cref{lem:bound-Xi-Pi} that 
\begin{align}\label{eqn:bound-adj}
	&\phantom{{}={}}\frac{1}{k\abs{\+C_{\tau,k}}}\Pi_\tau^{-\frac{1}{2}}\tp{\sum_{\omega\in\+C'_{\tau,k}}\Xi_\tau^{i,\omega}\tp{-(h-1)(a_{h-1}-a_h)\!{Adj}_\tau^i + 4a_h^2\gamma\tp{\!{Adj}_\tau^i}^2}\Xi_\tau^{i,\omega}}\Pi_\tau^{-\frac{1}{2}}\\\notag
	&\mle \frac{4\gamma h}{\beta-1} \!{Id}_\tau^i	\mle \frac{8h\cdot (\beta/(\beta-1)+\iota(\Delta))\exp\tp{\iota(\Delta)}}{(\beta-1)^2}\cdot\!{Id}_\tau^i.
\end{align}


A direct calculation yields that
\begin{equation}\label{eqn:bound-remainder}
	\frac{\Pi_\tau^{-\frac{1}{2}}\tp{\!{Remainder}}\Pi_\tau^{-\frac{1}{2}}}{k\abs{\+C_{\tau,k}}}\mle \frac{4\gamma}{\beta-1} \frac{2h^2}{\beta-1}\tp{1+\frac{2}{\beta-1}}\!{Id}_\tau^i \mle \tp{1+\frac{2}{\beta-1}}\frac{16 h \cdot (\beta/(\beta-1)+\iota(\Delta))\iota(\Delta)\exp\tp{\iota(\Delta)}}{(\beta-1)^2} \!{Id}_\tau^i.
\end{equation}
Combining \Cref{eqn:bound-adj} and \Cref{eqn:bound-remainder} finishes the proof.

\end{proof}

\subsubsection{Construction of $B_\tau$}
\label{sss:Btau}
For $\tau$ of co-dimension $k > 2$ such that $G_\tau$ is connected, we introduce coefficients  $\set{b'_h}_{1 \leq h \leq \Delta}$ and $\set{b_h}_{1 \leq h \leq \Delta}$ whose values will be determined later, and define $B_\tau$ as follows:
\begin{equation} \label{eqn:B-def}
	B_\tau(vc,vc) = 
	\begin{cases}
		b'_{\Delta_\tau(u)} & \text{if } \Delta_\tau(v)=1, u\sim_\tau v;\\
		b_{\Delta_\tau(v)} & \text{if } \Delta_\tau(v)\geq 2,
	\end{cases}
\end{equation}
for any $v \in V_\tau$ and all other entries are $0$ where $u\sim_\tau v$ means $u$ and $v$ are adjacent in $G_\tau$.
When $\Delta_\tau(v)=\Delta_\tau(u) = 1$, this is exactly the base case considered in \Cref{sec:base-case}. According to~\eqref{eqn:base-case-def-ol}, we have $b'_1=\frac1{\tp{\beta-1}^2}$. The definition of $B_\tau$ above for $\tau$ with connected $G_\tau$ extends to all $\tau$ by~\eqref{eqn:disconnected-A-B}.


Notice that we only need to satisfy the inductive constraint~\eqref{eqn:condition-main} when $G_\tau$ is connected. By \Cref{lem:matrix-product-2}, the constraint~\eqref{eqn:condition-main} is equivalent to 
\begin{align}\label{eqn:condition-main-2}
\begin{aligned}
 	\sum_{i\in \wh V} \Pi_\tau^{-1/2} \tp{(k-1)\cdot \E[x\sim\pi_\tau]{A_{\tau\cup\set{x}}^i} - (k-2) \cdot A_\tau^i+4A_\tau^i \Pi_\tau^{-1} A_\tau^i}\Pi_\tau^{-1/2} \\
 	 \mle
 	 (k-2)B_\tau -(k-1)\cdot \Pi_\tau^{-1}\E[x\sim\pi_\tau]{\Pi_{\tau\cup\set{x}} B_{\tau\cup\set{x}}}-2B_\tau^2.
 	 \end{aligned}
\end{align}
Denote the RHS of~\eqref{eqn:condition-main-2} by $B$, i.e., 
\[
B:=(k-2)B_\tau -(k-1)\cdot \Pi_\tau^{-1}\E[x\sim\pi_\tau]{\Pi_{\tau\cup\set{x}} B_{\tau\cup\set{x}}}-2B_\tau^2.
\]
Then we have
\begin{align}
\begin{aligned}
B(vc,vc)&=(k-2)B_\tau(vc,vc) -(k-1)\cdot \pi_\tau(vc)^{-1}\sum_{x\in \+C_\tau}{\pi_\tau(x)\pi_{\tau\cup \set{x}}(vc) B_{\tau\cup\set{x}}}(vc,vc)-2B_\tau(vc,vc)^2	\\
&=(k-2)B_\tau(vc,vc) -(k-1)\cdot \sum_{x\in \+C_\tau}{\pi_{\tau\cup \set{vc}}(x) B_{\tau\cup\set{x}}}(vc,vc)-2B_\tau(vc,vc)^2,
\end{aligned}
\label{eqn:contraint-B-entry}
\end{align}
where the last equality follows from the fact that $\pi_\tau(x)\pi_{\tau\cup \set{x}}(vc)=\pi_{\tau\cup\set{vc}}(x)\pi_\tau(vc)$. 
\paragraph{Case 1: $\Delta_\tau(v)=1$ and $u \sim_\tau v$.}
We have
\[
(k-1)\sum_{x \in \+C_\tau} \pi_{\tau\cup\set{vc}}(x) B_{\tau\cup\set{x}}(vc,vc)= (\Delta_\tau(u)-1)b'_{\Delta_\tau(u)-1}+(k-\Delta_\tau(u)-1)b'_{\Delta_\tau(u)}.
\]
Then 
\[
B(vc,vc)
=(\Delta_\tau(u)-1)(b'_{\Delta_\tau(u)}-b'_{\Delta_\tau(u)-1})-2(b'_{\Delta_\tau(u)})^2.
\]
\paragraph{Case 2: $\Delta_\tau(v)\geq 2$.}
Assume $b'_h\leq b_1$ for any $1\leq h \leq \Delta$. We have
\[
(k-1)\sum_{x \in \+C_\tau} \pi_{\tau\cup\set{vc}}(x) B_{\tau\cup\set{x}}(vc,vc)\leq \Delta_\tau(v)b_{\Delta_\tau(v)-1}+ (k-\Delta_\tau(v)-1)b_{\Delta_\tau(v)}.
\]
Then 
\[
B(vc,vc)
\geq (\Delta_\tau(v)-1)b_{\Delta_\tau(v)}-\Delta_\tau(v) b_{\Delta_\tau(v)-1}-2b_{\Delta_\tau(v)}^2.
\]
Also for any $v=(i_1,i_2)\in V_\tau$ where $i_1, i_2\in \wh V$, $\sum_{i\in \wh V\colon v\in V^i}h_\tau^i = h_\tau^{i_1}+h_\tau^{i_2}=\Delta_\tau(v)$. From the analysis in the above two cases, by \Cref{lem:bound-A-tau}, the constraint~\eqref{eqn:condition-main} is satisfied as long as
\begin{equation}\label{eqn:condition-Delta-one}
	\frac{8(\beta/(\beta-1)+\iota(\Delta))\exp\tp{\iota(\Delta)}}{(\beta-1)^2}  \tp{\frac{2(\beta+1)}{\beta-1}\iota(\Delta)+1} \le (\Delta_\tau(v)-1)(b'_{\Delta_\tau(v)}-b'_{\Delta_\tau(v)-1})-2(b'_{\Delta_\tau(v)})^2,
\end{equation}
and
\begin{equation}\label{eqn:condition-Delta}
	\frac{8(\beta/(\beta-1)+\iota(\Delta))\exp\tp{\iota(\Delta)}}{(\beta-1)^2}  \tp{\frac{2(\beta+1)}{\beta-1}\iota(\Delta)+1} \Delta_\tau(v) \le (\Delta_\tau(v)-1)b_{\Delta_\tau(v)}-\Delta_\tau(v) b_{\Delta_\tau(v)-1}-2b_{\Delta_\tau(v)}^2,
\end{equation}
for any $v \in V_\tau$ with $\Delta(v)\geq 2$. We remark that the above constraints (~\eqref{eqn:condition-Delta-one} and ~\eqref{eqn:condition-Delta}) only exist for vertices $v$ with $\Delta_\tau(v)\geq 2$ since $\!{codim}(\tau)>2$ and $G_\tau$ is connected. This is crucial since otherwise, the inequality system derived later has no solution. 
Assume $\beta \geq 11$. Since 
\begin{align}\label{eqn:A-tau-i-upper-bound}
\begin{split}
		\Pi_\tau^{-1/2}A_\tau^i\Pi_\tau^{-1/2}
		&= \frac{a_{h_\tau^i}}{k}\cdot\frac{1}{\abs{\+C_{\tau,k}}}\cdot \Pi_\tau^{-1/2}\sum_{\substack{\omega\in \+C_{\tau,k}\colon\\\forall w\in V^i_\tau\colon \omega(w)\ne c}} A^{i,\omega}_\tau\Pi_\tau^{-1/2} \\
		&= \frac{a_{h_\tau^i}}{k\abs{\+C_{\tau,k}}} \Pi_\tau^{-1/2}\sum_{\omega\in \+C_{\tau,k}'} \Xi_\tau^{i,\omega}(-\!{Adj}_\tau^i +\+R_\tau^{i,\omega})\Xi_\tau^{i,\omega}\Pi_\tau^{-1/2}\\
		&\mle \frac{a_{h_\tau^i}}{2(\beta-1)}\tp{1+\frac{2h_\tau^i}{\beta-1}}\!{Id}_\tau^i\\
		&\mle \frac{1}{2(\beta-1)}\frac{1+\frac{2h_\tau^i}{\beta-1}}{1+\frac{8(h_\tau^i-1)(1+\iota(\Delta))\exp(\iota(\Delta))}{\beta-1}}\!{Id}_\tau^i\\
		&\mle \frac{1}{20}\!{Id}_\tau^i,
\end{split}
\end{align}
we have $M_\tau=\frac{\sum_i A_\tau^i + \Pi_\tau B_\tau}{k-1} \mle \frac{k-1}{3k-1}\Pi_\tau$ as long as $B_\tau \mle \tp{\frac{(k-1)^2}{3k-1}-\frac{1}{10}}\!{Id}_\tau$.
We strengthen this constraint to $B_\tau \mle \tp{\frac{1}{5}-\frac{1}{10}}\!{Id}_\tau=\frac{1}{10}\!{Id}_\tau$.

For brevity, we donote $8(\beta/(\beta-1)+\iota(\Delta))\exp\tp{\iota(\Delta)}\tp{\frac{2(\beta+1)}{\beta-1}\iota(\Delta)+1}$ by $C(\Delta)$ in the following calculation. Therefore, our constraints for $\set{b'_h}_{1\leq h\leq \Delta}$ and $\set{b_h}_{1\leq h\leq \Delta}$ are
\begin{equation}
	\label{eqn:constraint-B}\tag{$\blacktriangle$}
	\begin{cases}
		b'_1 = \frac1{(\beta-1)^2}; &\\
		(h-1)(b'_h-b'_{h-1})\geq 2 (b'_h)^2 + \frac{C(\Delta)}{(\beta-1)^2} & 2\leq h \leq \Delta;\\
		b_1 \geq b'_h & 1\leq h \leq \Delta;\\
		(h-1) b_h - hb_{h-1} \ge 2 b_h^2+\frac{C(\Delta)h}{(\beta-1)^2},& 2\le h\le \Delta; \\
		b_h \le \frac1{10}, & 1\le h\le \Delta.
	\end{cases}
\end{equation}

It follows from \Cref{lem:solution-1} that there exists a feasible solution of $b'_h$ such that $b'_h \leq \frac{1+\tp{4C(\Delta) + 16/(\beta-1)^2}\log \Delta}{(\beta-1)^2}$ as long as $\beta \geq 4\sqrt{\log \Delta\tp{\sqrt{1+4C(\Delta)^2\log^2\Delta}+2C(\Delta)\log\Delta}}+1$. Since $C(\Delta)\geq 8$, the requirement can be strengthened to $\beta \geq 4\sqrt{4.01C(\Delta)}\log \Delta+1$ and $b'_h$ can be upper bounded by
\[
	b'_h\le \frac{4.2C(\Delta)\log\Delta}{(\beta-1)^2}.
\]
Therefore, we can set $b_1=\frac{4.2C(\Delta)\log\Delta}{(\beta-1)^2}$, then by \Cref{lem:solution}, there is a feasible solution of the \Cref{eqn:constraint-B} 
if $\beta \geq c(\sqrt{\Delta}\log^2\Delta+2)\sqrt{4.2C(\Delta)\log\Delta}+1$ where $c=2\sqrt{5(1+\frac{2}{4.2\log \Delta})}$.
Note that $C(\Delta)\leq 8\exp(\iota(\Delta))(2.4\iota(\Delta)^2+3.64\iota(\Delta)+1.1)$ when $\beta \geq 11$. Therefore, our final constraints for $\beta$ are
\begin{equation}
	\label{eqn:constraint-B-1}\tag{$\heartsuit$}
	\begin{cases}
		\beta \geq 21;\\
		\beta \geq \frac{\Delta}{\iota(\Delta)}+1;\\
		\beta \geq 4\sqrt{4.01\cdot 8\exp(\iota(\Delta))(2.4\iota(\Delta)^2+3.64\iota(\Delta)+1.1)}\log \Delta+1;\\
		\beta \geq 2\sqrt{5(1+\frac{2}{4.2\log \Delta})}(\sqrt{\Delta}\log^2\Delta+2)\sqrt{4.2\cdot8\exp(\iota(\Delta))(2.4\iota(\Delta)^2+3.64\iota(\Delta)+1.1)\log\Delta}+1.
	\end{cases}
\end{equation}
Taking $\iota(\Delta)=0.1\log \Delta$, we obtain the final requirement for $\beta$:
\[
\beta \geq \max\set{\frac{10\Delta}{\log \Delta}, 57(\Delta^{0.55}\log^{7/2}\Delta+2\Delta^{0.05}\log^{3/2}\Delta)}+1.
\]

\subsection{Proof of the Main Theorems}\label{sec:main-proof}
\begin{theorem}
\label{thm:mains}
Let $(\+C,\pi_n)$ be a weighted simplicial complex where $\pi_n$ is a uniform distribution over proper $\beta$-extra list-colorings over a line graph $G=(V,E)$ with $|V|=n$ and maximum degree $\Delta$.
Then as long as
\begin{equation*}
	 \beta \geq \max\set{\frac{10\Delta}{\log \Delta}, 57(\Delta^{0.55}\log^{7/2}\Delta+2\Delta^{0.05}\log^{3/2}\Delta)}+1 =O(\Delta/\log \Delta),
\end{equation*}
	we have $\lambda_2(P_\tau)\leq \frac{1}{9(k-1)}$ for any $\tau\in \+C$ of co-dimension $k$. Therefore
	\begin{enumerate}
		\item $\rho_{\!{LS}}(P_{\!{GL}}) = \Omega\tp{1/n}$, the mixing time of $P_{\!{GL}}$ is $O_\Delta(n \log n)$;
		\item $\lambda(P_{\!{GL}}) = \Omega\tp{n^{-10/9}}$, the mixing time of $P_{\!{GL}}$ is $O(n^{19/9}\log q)$,
	\end{enumerate}
	where $P_{\!{GL}}$ is the transition matrix of Glauber dynamics on $(\+C,\pi_n)$.
\end{theorem}
\begin{proof}[Proof of \Cref{thm:mains}]
	As we did in \Cref{sec:base-case} and \Cref{sec:induction}, we are able to construct a set of matrices $\set{M_\tau}_{\tau\in\+C_{\le n-2}}$\footnote{$\+C_{\le n-2}$ means all faces of dimension at most $n-2$ in $\+C$.} which are block diagonal with the block $N_\tau^c$ for each color $c$. In the following discussion, we fix a color $c$ and drop the superscript $c$. In our construction, $N_\tau=\frac{1}{k-1}(A_\tau+\Pi_\tau B_\tau)=\frac{1}{k-1} \tp{\sum_{i\in \wh V}A_\tau^i + \Pi_\tau B_\tau}$. By \eqref{eqn:constraint-B-1}, as long as 
	\begin{equation}
		\beta \geq \max\set{\frac{10\Delta}{\log \Delta}, 57(\Delta^{0.55}\log^{7/2}\Delta+2\Delta^{0.05}\log^{3/2}\Delta)}+1,
		\label{eqn:beta-final-bound}
	\end{equation}
	we have $B_\tau \mle \frac{1}{10}\!{Id}_\tau$. And by \eqref{eqn:A-tau-i-upper-bound} and \eqref{eqn:beta-final-bound}, we have $\Pi_\tau^{-1/2}A_\tau\Pi_\tau^{-1/2} \mle \frac{1}{2(\beta-1)} \sum_{i\in \wh V}\!{Id}_\tau^i \mle \frac{1}{\beta-1} \!{Id}_\tau\mle \frac{1}{90} \!{Id}_\tau$. Therefore, 
	$\lambda_1\tp{\Pi_\tau^{-1}N_\tau}\le \frac1{9(k-1)}$
	and hence $\set{M_\tau}$ satisfy all the conditions in \Cref{prop:mtd-inductive}. So we immediately have $\lambda_2\tp{P_\tau}\le \frac{1}{9(k-1)}$ by \Cref{prop:mtd-inductive}.
	
	Calculating the modified log-Sobolev constant by \Cref{prop:ltog}, we get
	$\rho_{\!{LS}}(P_{\!{GL}}) = O_\Delta(1/n)$, therefore the mixing time is $O_\Delta(n\log n)$, proving the first part of the theorem.

	Calculating the spectral gap, \Cref{prop:ltog-spectral} implies 
	\begin{align*}
		\log \lambda(P_{\!{GL}}) &\ge \log{\frac{1}{n}} + \sum_{k=2}^n \tp{\log\tp{k-\frac{10}{9}} - \log(k-1)}\\
		&\ge \log{\frac{1}{n}} + \log{\frac89} + \frac19 \tp{\tp{\sum_{k=2}^{n-1}\log(k-1) - \log(k)}}\\
		&\ge \log{\frac{8}{9n^{10/9}}},
	\end{align*}
	so $\lambda(P_{\!{GL}}) = \Omega\tp{n^{-10/9}}$.
	Since $\pi_{min}\ge q^{-n}$, by \Cref{lem:sg to mixing}, the mixing time is $O(n^{19/9}\log q)$, proving the second part of the theorem.
\end{proof}

\begin{proof}[Proof of \Cref{thm:main-informal-1} and \Cref{thm:main-informal-2}]
	Here we only need to unify the bound in \Cref{eqn:beta-final-bound} to the form of $\beta\ge C\cdot\frac\Delta{\log\Delta}$.
	By calculation, the maximum value of $57(\Delta^{0.55}\log^{7/2}\Delta+2\Delta^{0.05}\log^{3/2}\Delta)/(\Delta/\log\Delta)$ when $\Delta\ge 1$ is less than $20028$.
	So the final bound of $\beta$ is $\beta\ge 20028\cdot\frac{\Delta}{\log\Delta}$.
\end{proof}

\section{Solving the Constraints}
\newcommand*{\solC}{\eta{}}
\newcommand*{\solB}{p}
\newcommand*{\fcc}[1]{{\color{teal} #1}}
\label{sec:solution}

\begin{lemma}
	\label{lem:solution-1}
	Given positive constants $C_1, C_2$, and $H\in \mathbb N_{\ge 1}$,
	consider the inductive constraint
	\begin{equation}
		\label{eq:solution-1}\tag{$\blacktriangledown$}                     
		\begin{cases}
			b'_1 = \frac1{\solB^2}; &\\
			(h-1)\tp{b'_h - b'_{h-1}} \ge C_1 (b'_h)^2+\frac{C_2}{\solB^2}, & 2\le h\le H.
		\end{cases}
	\end{equation}
	Then \cref{eq:solution-1} is solvable when $\solB \geq 2\sqrt{2C_1\log H \tp{\sqrt{1+4C_2^2\log^2H}+2C_2\log H}}$.
	Moreover, under this condition, a feasible solution satisfies $b'_h\le \frac{1}{\solB^2}\tp{1+\tp{4C_2+\frac{8C_1}{\solB^2}}\log H}$.
\end{lemma}

\begin{proof}
	Let 
	\begin{equation*}
		b'_h = \frac1{\solB^2} \tp{1+\solC \log h},
	\end{equation*}
	We have $b'_1 = 1/\solB^2$ so the first constraint is satisfied.

	Plugging $b'_h$ into the second constraint, we have
	\begin{equation*}
		\frac{1}{\solB^2}(h-1){\solC \log \tp{\frac{h}{h-1}}} \ge C_1 \frac1{\solB^4}\tp{1 + \solC \log h}^2 + \frac{1}{\solB^2}C_2.
	\end{equation*}
	Since $(h-1)\log\tp{\frac{h}{h-1}}\ge \frac12$, $\tp{1+\solC \log h}^2\le 2 + 2\eta^2 \log^2 h$, $\log h\le \log H$, we can strengthen this constraint to
	\begin{equation*}
		\frac{1}{2\solB^2}\solC \ge C_1 \frac1{\solB^4}\tp{2 + 2(\solC \log H)^2} + \frac{1}{\solB^2}C_2.
	\end{equation*}
	Multiplying both sides by $2\solB^4$ and moving all terms to the same side, we have
	\begin{equation}
		 4 C_1 (\log^2H) \solC^2 - \solB^2\solC + 2{\solB^2}C_2 + 4C_1 \le 0.
		 \label{eq:sol01}
	\end{equation}
	
	As long as we can find an $\solC$ that satisfies \Cref{eq:sol01}, the constraint \Cref{eq:solution-1} is satisfied by our $b'_h$. Since \Cref{eq:sol01} is nothing but a quadratic equation about $\solC$, we can immediately calculate that the minimum solution of $\solC$ is
	\begin{align*}
		\solC^* &= \frac{\solB^2 - \sqrt{\solB^4 - 4\tp{16C_1^2\log^2 H + 8C_1C_2{\log^2 H} \solB^2}}}{8C_1\log^2 H}\\
		&\le  \frac{4\tp{16C_1^2\log^2 H + 8C_1C_2\log^2 H \solB^2}}{\solB^2\cdot8C_1\log^2 H}\\
		&= \frac{8C_1}{\solB^2} + 4C_2.
	\end{align*}
	Notice that $\solC^*$ exists as long as $\solB\ge \tp{64C_1^2\log^2H + \tp{32C_1C_2\log^2H}\solB^2}^{1/4}$, i.e., 
	\[
	\solB^2 \geq 8C_1\log H \tp{\sqrt{1+4C_2^2\log^2H}+2C_2\log H}.
	\]

	So by using $\solC^*$ as the coefficient $\eta$ in $b'_h$, we get a feasible solution of \Cref{eq:solution-1}, and for all $1\le h \le H$,
	\begin{equation*}
		b'_h\le b'_H\le \frac{1}{\solB^2}\tp{1+\solC\log H} \le \frac{1}{\solB^2}\tp{1+\tp{4C_2+\frac{8C_1}{\solB^2}}\log H}.
	\end{equation*}
\end{proof}

\begin{lemma}
\label{lem:solution}
Given positive constants $C_1, C_2, C_3$, $1/2\le \alpha\le 1$ and $H\in \mathbb N_{\ge 1}$,
consider the inductive constraint
\begin{equation}
	\label{eq:solution}\tag{$\star$}                     
	\begin{cases}
		b_1 = \frac1{\solB^2}; &\\
		(h-1) b_h - hb_{h-1} \ge C_1 b_h^2+\frac{C_2}{\solB^2} h^{2\alpha}, & 2\le h\le H;\\
		b_h \le C_3, & 1\le h\le H.
	\end{cases}
\end{equation}
Then \cref{eq:solution} is solvable when $\solB\ge cH^\alpha \log^2H + 2c$,
where
$c=\max\set{\sqrt{8C_1}, \sqrt{\frac{2}{C_3}}}\sqrt{1+2C_2}$.
\end{lemma}
\begin{proof}
	If $H$ = 1, then \cref{eq:solution} is equivalent to $\frac1{\solB^2}\le C_3$, it is satisfiable as long as $\fcc{\solB\ge \frac{1}{\sqrt{C_3}}}$.

	Next, we consider the case that $H\ge 2$.
    Let 
	\begin{equation*}
		b_h=\frac1{\solB^2}\tp{1+\solC\tp{H^{2\alpha-1}\log H}h\log h},
	\end{equation*}
	where $\solC > 0$ is a constant to be determined.
    Since $b_1 = \frac1{\solB^2}$, the first constraint holds.

	\paragraph{The second constraint:}
    Plugging $b_h$ into the second constraint, we have
    \begin{equation*}
		-\frac1{\solB^2}
        +\solC \frac{H^{2\alpha-1}\log H}{\solB^2}h(h-1)\log\tp{\frac h{h-1}}
		\ge
        \frac{C_1}{\solB^4}\tp{1 + \tp{\solC H^{2\alpha-1}\log H}h\log h}^{2}
        +\frac {C_2}{\solB^2}h^{2\alpha}
        \label{eq:sol1}
    \end{equation*}
    Multiplying $\solB^4$ on both sides, we have
    \begin{equation*}
		-\solB^2
        +\solB^2 \solC \tp{H^{2\alpha-1}\log H}h(h-1)\log\tp{\frac h{h-1}}
		\ge
        C_1\tp{1 + \tp{\solC H^{2\alpha-1}\log H}h\log h}^{2}
        +\solB^2 C_2h^{2\alpha}
        \label{eq:sol2}
    \end{equation*}
    Since $\log\tp{\frac{h}{h-1}}\ge \frac1h$ and $h(h-1)H^{2\alpha-1}\log H\ge0$,
    We can replace $\log \tp{\frac{h}{h-1}}$ for $\frac1h$, strengthening the constraint.
    \begin{equation}
        \solB^2\tp{\solC H^{2\alpha-1}\log H(h-1) - 1 - C_2h^{2\alpha}}
		-
        C_1\tp{1 + \tp{\solC H^{2\alpha-1}\log H}h\log h}^{2}
		\ge 0.
        \label{eq:sol3}
    \end{equation}

	Next, we simplify the two terms for calculation convenience.
	By the assumption $\alpha\ge 1/2$ and $h,H\ge 2$,
	we have $H^{2\alpha-1}\log H (h-1)\ge \log 2 > \frac12$.
	Moreover, $h\ge 2$ implies $h-1\ge \frac{h}2$, so for the first term, 
	\begin{equation}
        \solB^2\tp{\solC H^{2\alpha-1}\log H(h-1) - 1 - C_2h^{2\alpha}}
		\ge
		\solB^2\tp{\frac{\solC-2}{2} \tp{H^{2\alpha-1}\log H} h - C_2 h^{2\alpha}}.
		\label{eq:sol-t1}
	\end{equation}
	Here we assume $\solC-2\ge 0$ so that the above inequality holds.

	Similarly, for the second term, using $1 < 2\tp{H^{2\alpha-1}\log H}h\log h$,
	we have 
	\begin{equation}
        C_1\tp{1 + \tp{\solC H^{2\alpha-1}\log H}h\log h}^{2}
		\le 
		C_1\tp{\tp{2+\solC} \tp{H^{2\alpha-1}\log H} h\log h}^2.
		\label{eq:sol-t2}
	\end{equation}

	Now we can strengthen \cref{eq:sol3}
	using \cref{eq:sol-t1} and \cref{eq:sol-t2}. The constraint suffices to
	\begin{equation}
		\solB^2\tp{\frac{\solC-2}{2} \tp{H^{2\alpha-1}\log H} h - C_2 h^{2\alpha}}
		-
		C_1\tp{\tp{2+\solC} \tp{H^{2\alpha-1}\log H} h\log h}^2.
		\ge 0.
		\label{eq:sol4}
	\end{equation}

    Notice that since $\alpha\ge \frac12$, the left side is concave with respect to $h\ge 2$,
	so in order to prove our $b_h$ satisfy \cref{eq:sol3}, it suffices to prove \cref{eq:sol4}
	holds for $h=2$ and $h=H$.
    \begin{enumerate}
        \item When $h=2$, \cref{eq:sol4} is
			\begin{equation*}
			\solB^2\tp{\tp{\solC-2}\tp{H^{2\alpha-1}\log H} - C_2 2^{2\alpha}}
			-
			C_1\tp{\tp{2+\solC} \tp{H^{2\alpha-1}\log H} 2\log 2}^2
			\ge 0.
			\label{eq:sol11}
			\end{equation*}
			Since we have $C_2 2^{2\alpha}\le 4C_2\le 8C_2H^{2\alpha-1}\log H$,
			we can strengthen the constraint to
			\begin{align*}
				\solB^2\tp{\tp{\solC-2-8C_2}\tp{H^{2\alpha-1}\log H}}
				-
				C_1\tp{\tp{2+\solC} \tp{H^{2\alpha-1}\log H} 2\log 2}^2
				\ge 0.
			\end{align*}
			Here we additionally assume $\solC-2-8C_2 > 0$.
			As long as
			$\fcc{\solB\ge \frac{\sqrt{C_1}\tp{\solC+2}2\log 2}{\sqrt{\solC-2-8C_2}}
			\tp{H^{2\alpha-1}\log H}^{1/2}}$,
			the constraint holds for all $H\ge 2$.
        \item When $h=H$, \cref{eq:sol4} is 
			\begin{align*}
				\solB^2\tp{\frac{\solC-2}{2} H^{2\alpha}\log H - C_2 H^{2\alpha}}
				- C_1\tp{\tp{\solC+2}\tp{H^{2\alpha}\log^2 H}}^{2}
				\ge 0.
			\label{eq:sol6}
			\end{align*}
			Since $C_2H^{2\alpha}\le 4H^{2\alpha}\log H$, we can strengthen the constraint to
			\begin{align*}
				\solB^2\tp{\frac{\solC-2-8C_2}{2} H^{2\alpha}\log H}
				- C_1\tp{\tp{\solC+2}\tp{H^{2\alpha}\log^2 H}}^{2}
				\ge 0.
			\end{align*}
			as long as
			$\fcc{\solB\ge \frac{\sqrt{C_1}\tp{\solC+2}\sqrt2}{\sqrt{\solC-2-8C_2}}
			\tp{H^{2\alpha}\log^3 H}^{1/2}}$,
			the constraint holds for all $H\ge 2$.
    \end{enumerate}

	\paragraph{The third constraint:}
	We have
	\begin{equation*}
		b_h\le b_H = \frac{1}{\solB^2}\tp{1+\solC H^{2\alpha}\log^2 H}
		\le \frac{1}{\solB^2}\tp{\tp{\solC + 2} H^{2\alpha}\log^2 H}
	\end{equation*}
	This constraint holds for all $H\ge 2$ as long as 
	$\fcc{\solB\ge \sqrt{\frac{\tp{\solC+2}}{C_3}} H^\alpha\log H}$.

    \paragraph{Putting four bounds together:}
	From the above discussion, we get four bounds for $\solB$,
	such that as long as $\solB$ satisfies all of them, 
	\cref{eq:solution} holds for all $H\ge 1$.
	The four bounds are
	\begin{equation*}
		\begin{cases}
			\solB\ge \frac{1}{\sqrt{C_3}}\\
			\solB\ge \frac{\sqrt{C_1}\tp{\solC+2}2\log 2}{\sqrt{\solC-2-8C_2}}
			\tp{H^{2\alpha-1}\log H}^{1/2}\\
			\solB\ge \frac{\sqrt{C_1}\tp{\solC+2}\sqrt2}{\sqrt{\solC-2-8C_2}}
			\tp{H^{2\alpha}\log^3 H}^{1/2}\\
			\solB\ge \sqrt{\frac{\tp{\solC+2}}{C_3}}  H^\alpha\log H
		\end{cases}.
	\end{equation*}
	Letting $\solC = 6+16C_2$, the bounds become
	\begin{equation*}
		\begin{cases}
			\solB\ge \frac{1}{\sqrt{C_3}}\\
			\solB\ge 2\log 2\cdot4\sqrt{C_1\tp{1+2C_2}}
			\tp{H^{2\alpha-1}\log H}^{1/2}\\
			\solB\ge \sqrt{2}\cdot4\sqrt{C_1\tp{1+2C_2}}
			\tp{H^{2\alpha}\log^3 H}^{1/2}\\
			\solB\ge \sqrt{\frac{2\tp{1+2C_2}}{C_3}} H^\alpha\log H
		\end{cases}.
	\end{equation*}
	When $H\ge 3$, we have $\log H > 1$, so the four constraints unify to
	\begin{equation*}
		\solB\ge \max\set{\sqrt{8C_1}, \sqrt{\frac{2}{C_3}}}\sqrt{1+2C_2} H^\alpha\log^2 H.
	\end{equation*}
	When $H\le 2$, we have $H^\alpha\le2, \log H\le 1$, so the four constraints unify to
	\begin{equation*}
		\solB\ge 2\max\set{\sqrt{8C_1}, \sqrt{\frac{2}{C_3}}}\sqrt{1+2C_2}.
	\end{equation*}

	Finally, we proved as long as
	\begin{equation*}
		\solB\ge \max\set{\sqrt{8C_1}, \sqrt{\frac{2}{C_3}}}\sqrt{1+2C_2} H^\alpha\log^2 H
		+2\max\set{\sqrt{8C_1}, \sqrt{\frac{2}{C_3}}}\sqrt{1+2C_2},
	\end{equation*}
	\cref{eq:solution} is solvable.
\end{proof}
\bibliographystyle{alpha}
\bibliography{main.bib}

\newcommand{\etalchar}[1]{$^{#1}$}
\begin{thebibliography}{CLMM23}

\bibitem[AL20]{AL20}
Vedat~Levi Alev and Lap~Chi Lau.
\newblock Improved analysis of higher order random walks and applications.
\newblock In Konstantin Makarychev, Yury Makarychev, Madhur Tulsiani, Gautam
  Kamath, and Julia Chuzhoy, editors, {\em Proceedings of the 52nd Annual {ACM}
  {SIGACT} Symposium on Theory of Computing, {STOC} 2020, Chicago, IL, USA,
  June 22-26, 2020}, pages 1198--1211. {ACM}, 2020.

\bibitem[ALG21]{ALOG21}
Dorna Abdolazimi, Kuikui Liu, and Shayan~Oveis Gharan.
\newblock A matrix trickle-down theorem on simplicial complexes and
  applications to sampling colorings.
\newblock In {\em 62nd {IEEE} Annual Symposium on Foundations of Computer
  Science, {FOCS} 2021, Denver, CO, USA, February 7-10, 2022}, pages 161--172.
  IEEE, {IEEE}, 2021.

\bibitem[Bha97]{Bha97}
Rajendra Bhatia.
\newblock {\em Matrix analysis}.
\newblock Number 169 in Graduate texts in mathematics. Springer, New York,
  1997.

\bibitem[CDM{\etalchar{+}}19]{CDM+19}
Sitan Chen, Michelle Delcourt, Ankur Moitra, Guillem Perarnau, and Luke Postle.
\newblock Improved bounds for randomly sampling colorings via linear
  programming.
\newblock In {\em Proceedings of the Thirtieth Annual {ACM-SIAM} Symposium on
  Discrete Algorithms, {SODA} 2019, San Diego, California, USA, January 6-9,
  2019}, pages 2216--2234. SIAM, {SIAM}, 2019.

\bibitem[CG{\v{S}}V21]{CGSV21}
Zongchen Chen, Andreas Galanis, Daniel {\v{S}}tefankovi{\v{c}}, and Eric
  Vigoda.
\newblock Rapid mixing for colorings via spectral independence.
\newblock In D{\'{a}}niel Marx, editor, {\em Proceedings of the 2021 {ACM-SIAM}
  Symposium on Discrete Algorithms, {SODA} 2021, Virtual Conference, January 10
  - 13, 2021}, pages 1548--1557. {SIAM}, 2021.

\bibitem[CLMM23]{CLMM23}
Zongchen Chen, Kuikui Liu, Nitya Mani, and Ankur Moitra.
\newblock Strong spatial mixing for colorings on trees and its algorithmic
  applications.
\newblock In {\em 64th {IEEE} Annual Symposium on Foundations of Computer
  Science, {FOCS} 2023, Santa Cruz, CA, USA, November 6-9, 2023}, pages
  810--845. {IEEE}, 2023.

\bibitem[CLV21]{CLV21}
Zongchen Chen, Kuikui Liu, and Eric Vigoda.
\newblock Optimal mixing of glauber dynamics: entropy factorization via
  high-dimensional expansion.
\newblock In Samir Khuller and Virginia~Vassilevska Williams, editors, {\em
  {STOC} '21: 53rd Annual {ACM} {SIGACT} Symposium on Theory of Computing,
  Virtual Event, Italy, June 21-25, 2021}, pages 1537--1550. {ACM}, 2021.

\bibitem[DF03]{DF03}
Martin~E. Dyer and Alan~M. Frieze.
\newblock Randomly coloring graphs with lower bounds on girth and maximum
  degree.
\newblock {\em Random Struct. Algorithms}, 23(2):167--179, 2003.

\bibitem[DFHV13]{DFHV13}
Martin~E. Dyer, Alan~M. Frieze, Thomas~P. Hayes, and Eric Vigoda.
\newblock Randomly coloring constant degree graphs.
\newblock {\em Random Struct. Algorithms}, 43(2):181--200, 2013.

\bibitem[DHP20]{DHP20}
Michelle Delcourt, Marc Heinrich, and Guillem Perarnau.
\newblock The glauber dynamics for edge-colorings of trees.
\newblock {\em Random Struct. Algorithms}, 57(4):1050--1076, 2020.

\bibitem[FGYZ20]{FGYZ20}
Weiming Feng, Heng Guo, Yitong Yin, and Chihao Zhang.
\newblock Rapid mixing from spectral independence beyond the boolean domain.
\newblock {\em {ACM} Trans. Algorithms}, 18(3), Oct 2020.

\bibitem[GKM15]{GKM15}
David Gamarnik, Dmitriy Katz, and Sidhant Misra.
\newblock Strong spatial mixing of list coloring of graphs.
\newblock {\em Random Struct. Algorithms}, 46(4):599--613, 2015.

\bibitem[HJNP19]{MJNP19}
Marc Heinrich, Alice Joffard, Jonathan Noel, and Aline Parreau.
\newblock Unpublished manuscript.
\newblock
  \url{https://hoanganhduc.github.io/events/CoRe2019/CoRe_2019_Open_Problems.pdf},
  2019.

\bibitem[Jer95]{Jer95}
Mark Jerrum.
\newblock A very simple algorithm for estimating the number of k-colorings of a
  low-degree graph.
\newblock {\em Random Struct. Algorithms}, 7(2):157--165, 1995.

\bibitem[LPW17]{LPW17}
David~A Levin, Yuval Peres, and Elizabeth~L. Wilmer.
\newblock {\em Markov chains and mixing times}, volume 107.
\newblock American Mathematical Soc., 2017.

\bibitem[Mol04]{Mol04}
Michael Molloy.
\newblock The glauber dynamics on colorings of a graph with high girth and
  maximum degree.
\newblock {\em {SIAM} J. Comput.}, 33(3):721--737, 2004.

\bibitem[Opp18]{Opp18}
Izhar Oppenheim.
\newblock Local spectral expansion approach to high dimensional expanders part
  {I:} descent of spectral gaps.
\newblock {\em Discret. Comput. Geom.}, 59(2):293--330, Mar 2018.

\bibitem[SS97]{SS97}
Jes{\'u}s Salas and Alan~D Sokal.
\newblock Absence of phase transition for antiferromagnetic potts models via
  the dobrushin uniqueness theorem.
\newblock {\em Journal of Statistical Physics}, 86:551--579, 1997.

\end{thebibliography}

\appendix

\section{The Matrix Trickle-Down Theorem}\label{sec:proof-mtd}


Let $(\+C,\pi_n)$ be a weighted pure $n$-dimensional simplicial complex. 

\begin{proposition}[~\cite{Opp18}]
	\label{prop:garland}
	The following identities hold.
	\begin{enumerate}
		\item $\E[x\sim\pi_1]{\Pi_x}  = \Pi$.
		\item $\E[x\sim\pi_1]{\Pi_xP_x}  = \Pi P$.
		\item $\E[x\sim\pi_1]{\pi_x\pi_x^\top} = \Pi P^2$.
	\end{enumerate}
\end{proposition}
\begin{proof}
	\begin{enumerate}
		\item For any $y\in \+C_1$, it holds that
		\begin{align*}
			\E[x\sim\pi_1]{\Pi_x}(y,y)
			=\sum_{\substack{x\in\+C_1\\ x\ne y}} \pi_1(x)\pi_x(y)
			=\sum_{\substack{x\in\+C_1\\ x\ne y}} \frac{\Pr[\sigma\sim\pi_n]{x\in\sigma}}{n}\cdot\frac{\Pr[\sigma\sim\pi_n]{y\in \sigma\mid x\in\sigma}}{n-1}.
		\end{align*}
		This can be simplified to 
		\begin{align*}
			\sum_{\substack{x\in\+C_1\\ x\ne y}} \frac{\Pr[\sigma\sim\pi_n]{x,y\in\sigma}}{n(n-1)}=\frac{\Pr[\sigma\sim\pi_n]{y\in\sigma}}{n} = \Pi(y,y).
		\end{align*}
		\item For any $y,z\in\+C_1$, if $y=z$, $\E[x\sim\pi_1]{\Pi_xP_x}(y,z)  = \Pi P(y,z)=0$. Otherwise, direct calculation gives
		\begin{align*}
			\E[x\sim\pi_1]{\Pi_xP_x}(y,z)
			&=\sum_{\substack{x\in\+C_1\\x\ne y,z}}\pi_1(x)\pi_{x}(y)P_x(y,z)\\
			&=\sum_{\substack{x\in\+C_1\\x\ne y,z}}\frac{\Pr[\sigma\sim\pi_n]{x\in\sigma}}{n}\cdot \frac{\Pr[\sigma\sim\pi_n]{y\in \sigma\mid x\in\sigma}}{n-1}\cdot\frac{\Pr[\sigma\sim\pi_n]{y,z\in \sigma\mid x\in\sigma}}{(n-2)\cdot\Pr[\sigma\sim\pi_n]{y\in\sigma\mid x\in\sigma}}.
		\end{align*}
		This can be simplified to
		\begin{align*}
			\sum_{\substack{x\in\+C_1\\x\ne y,z}}\frac{\Pr[\sigma\sim\pi_n]{x,y,z\in\sigma}}{n(n-1)(n-2)}
			=\frac{\Pr[\sigma\sim\pi_n]{y,z\in\sigma}}{n(n-1)}.
		\end{align*}
		On the other hand, we have
		\[
		\Pi P(y,z) = \pi_1(y)\cdot \frac{\pi_2(y,z)}{2\pi_1(y)}=\frac{1}{2}\frac{\Pr[\sigma\sim\pi_n]{y,z\in\sigma} }{\binom{n}{2}}=\frac{\Pr[\sigma\sim\pi_n]{y,z\in\sigma}}{n(n-1)}.
		\]
		\item For every $y,z\in \+C_1$, it holds that
		\begin{align*}
			\E[x\sim\pi_1]{\pi_x\pi_x^\top}(y,z)
			&=\sum_{\substack{x\in\+C_1\\x\ne y,z}} \pi_1(x)\pi_x(y)\pi_x(z)\\
			&=\sum_{\substack{x\in\+C_1\\x\ne y,z}} \frac{\Pr[\sigma\sim\pi_n]{x\in\sigma}}{n}\cdot \frac{\Pr[\sigma\sim\pi_n]{y\in\sigma \mid x\in\sigma}}{n-1}\cdot\pi_x(z)\\
			&=\sum_{\substack{x\in\+C_1\\x\ne y,z}} \frac{\Pr[\sigma\sim\pi_n]{x,y\in \sigma}}{n(n-1)}\cdot \pi_x(z).
		\end{align*}
		On the other hand, note that for every $x\in\+C_1$, the row of $P$ indexed by $x$ is $\pi_x$, we have
		\begin{align*}
			\tp{\Pi P^2}(y,z)
			=\sum_{x\in\+C_1}(\Pi P)(y,x)\cdot P(x,z)
			=\sum_{x\in\+C_1} \pi_1(y)\pi_y(x)\cdot \pi_x(z).
		\end{align*}
		This can be written as
		\begin{align*}
			\sum_{\substack{x\in\+C_1\\x\ne y,z}} \frac{\Pr[\sigma\sim\pi_n]{y\in\sigma}}{n}\cdot \frac{\Pr[\sigma\sim\pi_n]{x\in\sigma \mid y\in\sigma}}{n-1} \cdot \pi_x(z)
			=\sum_{\substack{x\in\+C_1\\x\ne y,z}} \frac{\Pr[\sigma\sim\pi_n]{x,y\in \sigma}}{n(n-1)}\cdot \pi_x(z).
		\end{align*}

	\end{enumerate}
\end{proof}

To prove \Cref{prop:mtd}, we introduce the following property of the Loewner order.
\begin{lemma}[Lemma 2.3 in \cite{ALOG21}]
	\label{lem:monotone}
	Let $A, B\in \bb R^{n\times n}$ be two symmetric matrices. If $A(I-\eps A)\mle B(I-\eps B)$ for a constant $\eps>0$ and  $A, B \mle \frac{1}{2\eps}I$, then $A\mle B$.
\end{lemma}
\begin{proof}
	Consider the matrix function $f(M)=M(I-M)$ for a symmetric matrix $M \in \bb R^{n\times n}$ and $M\mle \frac{1}{2\eps}I$. If $M=\sum_{i=1}^n \mu_i v_i v_i^{\top}$, then we have $M(I-M)=\sum_{i=1}^n \mu_i(1-\eps \mu_i) v_i v_i^{\top}$, which means that $f$ is a matrix function extended from a real bijective function $x \mapsto x(1-\eps x)$ from $(-\infty, 1/2\eps]$ to $(-\infty, 1/4\eps^2]$. Therefore, from its inverse function $x \mapsto \frac{1}{2\eps}-(\frac{1}{4\eps^2}-x)^{1/2}$, we give the inverse function of $f$ on $\set{M\in \bb R^{n\times n} \cmid M\mle \frac{1}{4\eps^2}I}$:
	\[
	f^{-1}(M)=\frac{1}{2\eps}I-\left(\frac{1}{4\eps^2}-M\right)^{1/2}.
	\]
	Since $M\mapsto \sqrt{M}$ is monotone under Loewner order (see e.g. Theorem V.1.9 of \cite{Bha97}), it can be obtained by simple calculation that $f^{-1}$ is monotone under Loewner order.
\end{proof}
\begin{proof}[Proof of \Cref{prop:mtd}]
	For every $x\in\+C_1$, \cref{eqn:Pxcond} is equivalent to 
	\[
		\Pi_x P_x-\alpha \pi_x\pi_x^\top \mle \Pi_x N_x\mle \frac{1}{2\alpha+1}\Pi_x.
	\]
	Taking expectation and applying \Cref{prop:garland}, we have
	\[
		\Pi P-\alpha\Pi P^2\mle \E[x\sim \pi]{\Pi_x N_x}\mle \frac{1}{2\alpha+1}\Pi,
	\]
	which is equivalent to
	\[
		P-\alpha P^2\mle_{\pi}\Pi^{-1}\E[x\sim \pi]{\Pi_x N_x} \mle_{\pi} \frac{1}{2\alpha+1} I.
	\]
	Therefore, $P-\alpha P^2 \mle_{\pi} N-\alpha N^2$. Picking $Q=P-\beta\*1\pi^\top$ with $\beta=2-\frac{1}{\alpha}$, we immediately have 
	\[
		P-\alpha P^2 = Q-\alpha Q^2
	\]
	by comparing the spectrum. As a result,
	\[
		Q-\alpha Q^2\mle_\pi N-\alpha N^2,
	\]
	Since $\lambda_2(P_x)\leq \frac{1}{2\alpha+1}$, by the original trickle-down theorem (\Cref{prop:trickle-down}), we have $\lambda_2(P)\leq \frac{1}{2\alpha}$. Combined with $\beta=2-\frac{1}{\alpha}\geq 1-\frac{1}{2\alpha}$  we obtain that $Q\mle \frac{1}{2\alpha}I$. It follows from \Cref{lem:monotone} that $Q\mle_\pi N$.
\end{proof}


\section{Marginal Probability Bounds in \cite{GKM15}}
\label{sec:marginal-bound}
\begin{proof}[Proof of \Cref{lem:marginal-bound}]
	Let $V'=V_\tau \setminus{u}$. For any proper partial coloring $\omega$ on $V'$ (boundary) when $\tau$ is pinned, let $p_{G,\+L}(uc | \omega)=\Pr[\sigma \sim \mu]{\sigma(u)=c \mid \tau\cup \omega \subset \sigma}$ where $\mu$ is the uniform distribution over all proper colorings of $(G,\+L)$ and $p_{G,\+L}(uc)= \mu_u^\tau(c)$. 
	
	For the upper bound of $\mu_u^\tau(c)$, it holds that
	\begin{equation}
		p_{G,\+L}(uc |\omega)\leq \frac{1}{\ell_u^{\tau\cup \omega}} \leq \frac{1}{\ell_u^{\tau}-\Delta_\tau(u)}.
	\end{equation}
	Therefore,
	\begin{equation}
		\label{eqn:marginal-upper}
	\mu_u^\tau(c)=p_{G,\+L}(uc) = \sum_{\omega} p_{G,\+L}(uc |\omega) \mu_{V'}^\tau(\omega) \leq \frac{1}{\ell_u^{\tau}-\Delta_\tau(u)} \leq \frac{1}{\beta}.
	\end{equation}
	As for the lower bound, assuming the free neighbors of $u$ after pinning $\tau$ are $v_1, v_2, \dots, v_{\Delta_\tau(u)}$, recall the recursion on marginal probabilities:
	\[
	p_{G,\+L}(uc)=\frac{\prod_{i=1}^{\Delta_\tau(u)} (1-p_{G_{u},\+L_{i,c}}(v_i c))}{\sum_{c'\in L_u^\tau}\prod_{i=1}^{\Delta_\tau(u)} (1-p_{G_{u},\+L_{i,c'}}(v_i c'))},
	\]
	where $(G_{u},\+L_{i,c})$ denote the coloring instance after removing the vertex $u$ and color $c$ from the color lists of $u$'s neighbors $v_j$ for $j<i$.
	From \eqref{eqn:marginal-upper} we have $\prod_{i=1}^{\Delta_\tau(u)}(1- p_{G_{u},\+L_{i,c}}(v_i c)) \geq (1-\beta^{-1})^{\Delta_\tau(u)}$. Also $\sum_{c'\in L_u^\tau}\prod_{i=1}^{\Delta_\tau(u)} (1-p_{G_{u},\+L_{i,c'}}(v_i c'))\leq \ell_u^\tau$. So we have
	\[
	\mu_u^\tau(c)=p_{G,\+L}(uc) \geq \frac{(1-\beta^{-1})^{\Delta_\tau(u)}}{\ell_u^{\tau}}.
	\]
\end{proof}
\end{document}